\documentclass[11pt,reqno]{amsart}

\usepackage{amscd,amssymb,amsmath,amsthm}
\usepackage{graphicx}
\usepackage{color}
\usepackage{cite}
\newtheorem{thm}{Theorem}

\newtheorem{lemma}{Lemma}
\newtheorem{pro}{Proposition}
\newtheorem{rk}{Remark}

\numberwithin{equation}{section} 

\def\s{\sigma}

\def\s{\sigma}

\def\s{\sigma}

\def\Z{\mathbb{Z}}

\begin{document}
\title[Gibbs measures for the SOS model]{
Extremality of translation-invariant phases 
for a finite-state SOS-model on the binary tree 
}

\author{C. Kuelske, U. A. Rozikov}

\address{C.\ Kuelske\\ Fakult\"at f\"ur Mathematik,
Ruhr-University of Bochum, Postfach 102148,\,
44721, Bochum,
Germany}
\email {Christof.Kuelske@ruhr-uni-bochum.de}

\address{U.\ A.\ Rozikov\\ Institute of mathematics,
29, Do'rmon Yo'li str., 100125, Tashkent, Uzbekistan.}
\email {rozikovu@yandex.ru}

\begin{abstract} 
We consider the SOS (solid-on-solid) model, with spin values $0,1,2$, 
on the Cayley tree of order two (binary tree). 
We treat both ferromagnetic and antiferromagnetic coupling,  
with interactions which are proportional to the absolute value of the spin differences. 

We present a classification of all 
translation-invariant phases (splitting Gibbs measures) of the model: 
We show  uniqueness in the case of 
antiferromagnetic interactions, and existence of up to seven phases in the case of ferromagnetic 
interactions, where the number of phases depends on the interaction strength. 

Next we investigate whether these 
states are extremal or non-extremal in the set of all Gibbs measures, 
when the coupling strength is varied,  whenever they exist.  
 We show that two states are always extremal,   
two states are always non-extremal, 
while three of the seven states make transitions between extremality and non-extremality. 
We provide explicit bounds on those transition values,  
making use of algebraic properties of the models, and an adaptation of the method of 
Martinelli, Sinclair, Weitz.  

\end{abstract}
\maketitle



\section{Introduction}

A solid-on-solid (SOS) model is a spin system with spins taking values in (a subset of) the 
integers, and formal Hamiltonian 
$$
 H(\sigma)=-J\sum_{\langle x,y\rangle}
|\sigma(x)-\sigma(y)|,
$$
where $J\in \mathbb{R}$ is a coupling constant.  
An (infinite-volume) spin-configuration $\sigma$ is a function 
from the vertices of the underlying graph
 to the local configuration space $\Phi \subset \Z$. The vertex will be the Cayley tree in 
 our case, and for most of our analysis we will restrict to the binary tree.  As usual,
$\langle x,y\rangle$ denotes a pair of nearest neighbor vertices.  
For the local configuration space $\Phi$ we consider in the present paper the {\em finite set}  
$\Phi:=\{0,1,\ldots , m\}$, where $m\geq 1$. Most of the times we will further specify to $m=2$ for which 
we will present an (almost) complete analysis of the translation-invariant.

The model can be considered as a generalization
of the Ising model, which corresponds to $m=1$, or a less symmetric variant of the Potts model.  
SOS-models on the cubic lattice were analyzed in \cite{Maz} where an analogue of
the so-called Dinaburg--Mazel--Sinai theory was developed.
Besides interesting phase transitions in these models, the
attention to them is motivated by applications, in particular
in the theory of communication networks; see, e.g., \cite{Kel}, \cite{Ra}.

SOS models with $\Phi=\Z$ have been used as simplified discrete 
interface models which should approximate the behavior of a Dobrushin-state 
in an Ising model when the underlying graph is $\Z^d$, and $d\geq 2$. 
While there is the issue of possible 
non-existence of Gibbs-states in the case of such unbounded spins, in particular in the additional presence of {\em  disorder} (see \cite{BK} and \cite{BK1}),  this issue is not present here, and we are looking for a classification of the phases. 

Indeed, compared to the Potts model, the SOS model has
less symmetry: The full symmetry of the Hamiltonian 
under joint permutation of the spin values is reduced to the mirror symmetry, 
which is the invariance of the model under the map $\s_i \mapsto m- \s_i$ on the local spin space.  
Therefore one expects a more diverse structure
of phases. Note that, in the ferromagnetic  case
it is intuitively plausible that the ground states
corresponding to `middle-level
surfaces' will be `dominant' as they carry more entropy. 
 This observation was made formal
in \cite{Maz} for the model on a cubic lattice.

To the best of our knowledge, the first paper devoted to the SOS model on the Cayley tree is \cite{Ro12}.
In \cite{Ro12} the case of  arbitrary $m\geq 1$ is treated and 
a vector-valued functional equation for possible boundary laws of the model is obtained. 
Recall that each solution to this functional equation determines a splitting Gibbs measure (SGM), 
in other words a tree-indexed Markov chain. Such measures can be obtained by propagating spin values along the edges of the tree, from root to the outside, 
with a transition matrix depending on initial 
Hamiltonian and the boundary law solution. 
In particular the constant (site-independent) boundary laws then 
define translation-invariant (TI) SGMs. 

Also the symmetry (or absence of symmetry) of the Gibbs measures under spin reflection is seen in terms of the corresponding boundary law.  
TISGM's which are symmetric have already been studied in SOS models 
in the particular cases of $m=2$  in \cite{Ro12}, and   $m=3$  in   \cite{Ro12}. See also \cite{Ro} for more details about SOS models on trees. 

However, the study of TISGMs which are {\em not mirror symmetric} is new.  
In this paper we describe all TISGMs (including (non-)symmetric ones) of the three-state ($m=2$) SOS model on the Cayley tree of order two.  

The paper is organized as follows. Section 2 contains preliminaries (necessary definitions and facts) and the main result of this paper.
In Section 3 we shall give the description of all TISGMs, 
and show that their number can be up to seven, at any given 
value of the coupling. We then turn to the question of their extremality. 
Analogous questions has been studied by the authors for all TISGMs of the
Potts model in \cite{KRK}, \cite{KR}
and we will draw from our experience  to treat the present situation, incorporating the non-symmetric states. As we will see, our classification for the SOS-model leaves 
fewer gaps than for the Potts model. 
More precisely, Subsection 4.1 is devoted to conditions implying 
the non-extremality for each such TISGM. 
We shall investigate whether and for which phases and temperatures 
the Kesten-Stigum condition \cite{Ke} for the second largest eigenvalue of the transition 
matrix holds.  
 Subsection 4.2 is then devoted to the converse problem of giving conditions for extremality of 
 TISGMs in our model.  
Here we use the approach of Martinelli, Sinclair, Weitz  \cite{MSW} to derive our bounds 
on the parameter regimes for the absence of reconstruction solvability (extremality). 
  
\section{Preliminaries and the main result}

{\it Cayley tree.}
The Cayley tree $\Gamma^k$
of order $ k\geq 1 $ is an infinite tree, i.e., a graph without
cycles, such that exactly $k+1$ edges originate from each vertex.
Let $\Gamma^k=(V, L)$ where $V$ is the set of vertices and  $L$ the set of edges.
Two vertices $x$ and $y$ are called {\it nearest neighbors} if there exists an
edge $l \in L$ connecting them.
We will use the notation $l=\langle x,y\rangle$.
A collection of nearest neighbor pairs $\langle x,x_1\rangle, \langle x_1,x_2\rangle,...,\langle x_{d-1},y\rangle$ is called a {\it
path} from $x$ to $y$. The distance $d(x,y)$ on the Cayley tree is the number of edges of the shortest path from $x$ to $y$ (which is the unique path if no edges are crossed twice). 

For a fixed $x^0\in V$, called the root, we set
\begin{equation*}
W_n=\{x\in V\,| \, d(x,x^0)=n\}, \qquad V_n=\bigcup_{m=0}^n W_m
\end{equation*}
and denote by
$$
S(x)=\{y\in W_{n+1} :  d(x,y)=1 \}, \ \ x\in W_n, $$ the set  of {\it direct successors} of $x$.

{\it SOS model.} We consider models where the spin takes values in the set
$\Phi:=\{0,1,\ldots , m\}$, $m\geq 2$, and is assigned to the vertices
of the tree. A configuration $\sigma$ on $V$ is then defined
as a function $x\in V\mapsto\sigma (x)\in\Phi$;
the set of all configurations is $\Phi^V$.
The (formal) Hamiltonian is of an SOS form:
\begin{equation}\label{rs1.1}
 H(\sigma)=-J\sum_{\langle x,y\rangle\in L}
|\sigma(x)-\sigma(y)|,
\end{equation}
where $J\in \mathbb{R}$ is a coupling constant.

Here, $J<0$
gives a ferromagnetic  and $J>0$ an anti-ferromagnetic
model.

We use a standard definition  of a Gibbs measure (which is an infinite-volume measure 
which satisfies the DLR equation), and of a translation-invariant (TI)
measure (which is a measure which is invariant under translations which map the 
tree onto itself). Also, we call measure $\mu$ {\em symmetric} 
if it is preserved under the simultaneous change $j\mapsto m-j$
at each vertex $x\in V$.

{\it Functional equations and splitting Gibbs measures.}
Now we shall give a system of functional equations for boundary laws $z$ (or equivalently 
boundary fields $h$) whose solutions correspond
to Gibbs measures of SOS model on the Cayley tree. Every extremal Gibbs measure arises 
in this way (even without the requirement of 
translation-invariance), but not necessarily every measure which arises in this way is extremal, 
see\cite{Ge}. 
We recall the derivation of the equations via the compatibility requirement 
for the convenience of the reader. 

Let $h:\;x\mapsto h_x=(h_{0,x}, h_{1,x},...,h_{m,x})
\in \mathbb{R}^{m+1}$ be a real vector-valued function of $x\in V\setminus
\{x^0\}$, assigning to the vertex $x$ a boundary field (depending on the $m+1$ 
different spin-values in the local spin space $\Phi$.)

 Given $n=1,2,\ldots$,
consider the probability distribution $\mu_n$ on
$\Phi^{V_n}$ defined by
\begin{equation}\label{rs2.1}
\mu^{(n)}(\sigma_n)=Z_n^{-1}\exp\left(-\beta H(\sigma_n)
+\sum_{x\in W_n}h_{\sigma(x),x}\right).
\end{equation}

Here, $\sigma_n:x\in V_n\mapsto \sigma(x)$
and $Z_n$ is the corresponding partition function:

\begin{equation}\label{rs2.2}
Z_n=\sum_{{\widetilde\sigma}_n\in\Phi^{V_n}}
\exp\left(-\beta H({\widetilde\sigma}_n)
+\sum_{x\in W_n}h_{{\widetilde\sigma}(x),x}\right).
\end{equation}

We say that the probability distributions $\mu^{(n)}$
are compatible if $\forall$ $n\geq 1$ and $\sigma_{n-1}\in\Phi^{V_{n-1}}$:
\begin{equation}\label{rs2.3}
\sum_{\omega_n\in\Phi^{W_n}}\mu^{(n)}(\sigma_{n-1}\vee\omega_n)=
\mu^{(n-1)}(\sigma_{n-1}).
\end{equation}
Here $\sigma_{n-1}\vee\omega_n\in\Phi^{V_n}$ is the concatenation
of $\sigma_{n-1}$ and $\omega_n$.
In this case there exists a unique measure $\mu$ on
$\Phi^V$ such that, $\forall$ $n$ and
$\sigma_n\in\Phi^{V_n}$, $\mu \left(\left\{\sigma
\Big|_{V_n}=\sigma_n\right\}\right)=\mu^{(n)}(\sigma_n)$. Such
a measure is called a {\it splitting Gibbs measure} (SGM) corresponding to
Hamiltonian $H$ and function $x\mapsto h_x$, $x\neq x^0$.

The following statement describes the conditions on
the boundary fields $h_x$ guaranteeing
compatibility of distributions $\mu^{(n)}(\sigma_n).$
When we do this we can reduce the dimension by one 
since boundary fields (which act as energies in the exponent) 
are defined only up to additive constants.

\begin{pro}\label{rsp2.1}\cite{Ro12} Probability distributions
$\mu^{(n)}(\sigma_n)$, $n=1,2,\ldots$, in (\ref{rs2.1}) are compatible iff for any $x\in V\setminus\{x^0\}$
the following equation holds:
\begin{equation}\label{rs2.4}
 h^*_x=\sum_{y\in S(x)}F(h^*_y,m,\theta).
 \end{equation}
Here,
\begin{equation}\label{rs2.5}
\theta=\exp(J\beta ),
\end{equation}
$h^*_x$ stands for the vector
$(h_{0,x}-h_{m,x}, h_{1,x}-h_{m,x},...,h_{m-1,x}-h_{m,x})$ and
the vector function $F(\;\cdot\;,m,\theta ):\;\mathbb{R}^m\to \mathbb{R}^m$
is $F(h,m,\theta )=
(F_0(h,m,\theta ),\ldots ,F_{m-1}(h,m,\theta))$, with
\begin{equation}\label{rs2.6}
F_i(h,m,\theta )=\ln{\sum_{j=0}^{m-1}
\theta^{|i-j|}\exp(h_j)+\theta^{m-i}\over
\sum_{j=0}^{m-1}\theta^{m-j}\exp(h_j)+1},
\end{equation}
$h=(h_0,h_1,...,h_{m-1}),
i=0,\ldots ,m-1.$
\end{pro}

From Proposition \ref{rsp2.1} it follows that for any $h=\{h_x,\ \ x\in V\}$
satisfying (\ref{rs2.4}) there exists a unique SGM $\mu$ for SOS model. However,
the analysis of solutions to (\ref{rs2.4}) for an arbitrary
$m$ is not easy.

{\it Translation-invariant SGMs.}
It is natural to begin with TI solutions
where $h_x=h\in \mathbb{R}^m$ is constant vector. In this case the equation  (\ref{rs2.4}) becomes
\begin{equation}\label{ti0}
z_i=\left(\sum_{j=0}^{m-1}
\theta^{|i-j|}z_j+\theta^{m-i}\over
\sum_{j=0}^{m-1}\theta^{m-j}z_j+1\right)^k, \ \ i=0,\ldots ,m-1,
\end{equation}
where $z_i=\exp(h_i)$. The vector $(z_0, \dots, z_{m-1})$ is 
called a (translation-invariant) law.  More generally it is common, also 
in the non-translation invariant case,  to call the exponentials of boundary fields 
the boundary laws.

\begin{rk} The system of equations (\ref{ti0}) has parameters $k\geq 2$, $m\geq 2$ and $\theta>0$, 
and it seems very difficult to find all solutions in the general case. 

In cases $m=2$ and $m=3$ the existence of mirror symmetric solutions (i.e. with  $z_{m-j}=z_j$, $j=0,1,\dots,m$) to the system (\ref{ti0}) were studied in \cite{Ro12} and \cite{Ro13}.
In this paper our goal is to give full analysis of
solutions of the system (\ref{ti0}) for $k=2$ and $m=2$.
We shall prove that in this case the system has up to seven solutions.
Moreover we will find explicit formulas of the solutions, which we then use to check the (non-)extremality of the corresponding Gibbs measures.
\end{rk}

{\it The main result.} The following theorem is the main result of this paper
\begin{thm}\label{tm}
For the SOS model with $m=2$ on the Cayley tree of order two the following assertions hold:
 there exist $\theta_c$ $(\approx 0.1414)$ and $\theta_c'$ $(\approx 0.2956$\footnote{see (\ref{tc}) for an exact value}) such that
 
 I.(Existence)
\begin{itemize}
 \item[1)] If $\theta>\theta_c'$ then there exists a unique TISGM $\mu_1$;
  \item[2)] If $\theta=\theta_c'$ then there are exactly three TISGMs $\mu_i$, $i=1,4,6$;
  \item[3)] If $\theta_c<\theta<\theta_c'$ then there are exactly five TISGMs $\mu_i$, $i=1,4,5,6,7$;
  \item[4)] If $\theta=\theta_c$ then there are exactly six such measures $\mu_i$, $i=1,3,4,5,6,7$;
  \item[5)] If $\theta<\theta_c$ then there are exactly seven such measures $\mu_i$, $i=1,2,3,4,5,6,7$.
    \end{itemize}
   II.(Extremality)
   \begin{itemize}
   \item[a)] There are values $\bar{\theta}$ $(\approx 2.655)$ and $\bar{\bar{\theta}}$ $(\approx 2.8765)$ such that the measure $\mu_1$ is extreme if $\theta<\bar{\theta}$ and is non-extreme if $\theta>\bar{\bar{\theta}}$.
   \item[b)] The measures $\mu_2$ and $\mu_3$ are non-extreme (where they exist).
   \item[c)] There are values $\theta^*$ $(\approx 0.17172)$ and $\theta^{**}$ $(\approx 0.26586)$ such that the measures $\mu_5$ and $\mu_6$ are non-extreme if $\theta<\theta^*$ and are extreme if $\theta>\theta^{**}$.
   \item[d)] The measures $\mu_4$ and  $\mu_7$ are extreme (where they exist).
   \end{itemize}
      \end{thm}
      
  We shall find all solutions for the translation-invariant boundary laws in our model and 
  prove part I of Theorem \ref{tm}  
 in Section 3.  He we are helped by the nature of the binary tree which helps 
 to keep the order of polynomials which need to be solved bounded by $4$. 
   In Subsection 4.1 we shall give results concerning to non-extremality and in Subsection 4.2 we give conditions of extremality.  

\section{Case $k=m=2$: Full analysis of solutions}

Assuming $k=m=2$ the two-dimensional fixed point 
equation (\ref{ti0}) for the two components of the boundary law can be written
in terms of the convenient variables 
$x=\sqrt{z_0}$ and $y=\sqrt{z_1}$ in the form 

\begin{equation}\label{rs3.2a}
x={x^2+\theta y^2+\theta^2 \over \theta^2x^2+\theta y^2+1},
\end{equation}
\begin{equation}\label{rs3.2b}
 y={\theta x^2+y^2+\theta \over \theta^2x^2+\theta y^2+1}
\end{equation}
From the equation (\ref{rs3.2a}) we get $x=1$ or
\begin{equation}\label{y}
\theta y^2=(1-\theta^2)x-\theta^2(x^2+1).
\end{equation}
\begin{rk} \label{<1} Since $x>0$ we have that the equality (\ref{y}) can hold iff $\theta<1$.
\end{rk}
\subsection{Case: $x=1$.} In this case from the equation (\ref{rs3.2b}) we get
\begin{equation}\label{y3}
\theta y^3-y^2+(\theta^2+1)y-2\theta=0.
\end{equation}
Using Cardano's formula one can prove the following

\begin{lemma}\label{l1}
There exists a unique $\theta_c (\approx 0.1414)$ such that
\begin{itemize}
\item If $\theta<\theta_c$ then the equation (\ref{y3}) has three solutions $y_3<y_2<y_1$ which are positive.
\item If $\theta=\theta_c$ then the equation has two positive solutions $y_2<y_1$.
\item If $\theta>\theta_c$ then the equation has one solution $y_1>0$.
\end{itemize}
\end{lemma}

\subsection{Case: $x\ne 1$ and (\ref{y}) is satisfied.} By Remark \ref{<1} we should only consider the case $\theta<1$.
The equation (\ref{rs3.2b}) can be written as 
\begin{equation}\label{2b}
 y^2=\left({\theta x^2+y^2+\theta \over \theta^2x^2+\theta y^2+1}\right)^2
\end{equation}
In the case when the equality (\ref{y}) is satisfied
from the equation (\ref{2b}) we get
$$((1-\theta^2)x-\theta^2(x^2+1))\theta=\left({x\over x+1}\right)^2,$$
which is equivalent to
\begin{equation}\label{x4}
\theta^3 x^4+\theta(3\theta^2-1)x^3+(4\theta^3-2\theta+1)x^2+\theta(3\theta^2-1)x+\theta^3=0.
\end{equation}
Denoting $\xi=x+1/x$ from (\ref{x4}) we get
\begin{equation}\label{xi}
\theta^3\xi^2+\theta(3\theta^2-1)\xi+2\theta^3-2\theta+1=0.
\end{equation}
This equation has no solution if $D=\theta^2(\theta-1)(\theta^3+\theta^2+3\theta-1)<0$; it has a 
unique solution if $D=0$ and two solutions if $D>0$.

For $\theta<1$ we note that $D=0$ has a unique solution:
\begin{equation}
\label{tc}
\theta_c'={1\over 3}\left(\sqrt[3]{26+6\sqrt{33}}-{8\over \sqrt[3]{26+6\sqrt{33}}}-1\right) \approx 0.2956.
\end{equation}
Thus we have the following
\begin{itemize}
\item If $\theta\in(0,\theta_c')$ then the equation (\ref{xi}) has two solutions $\xi_1<\xi_2$ with
$$\xi_{1}={1-3\theta^2-\sqrt{(\theta-1)(\theta^3+\theta^2+3\theta-1)}\over 2\theta^2},  \ \ \xi_{2}={1-3\theta^2+\sqrt{(\theta-1)(\theta^3+\theta^2+3\theta-1)}\over 2\theta^2};$$
\item If $\theta=\theta_c'$ then the equation (\ref{xi}) has a unique solution $\xi_1={1-3\theta^2\over 2\theta^2}$;
\item If $\theta\in(\theta_c',1)$ then the equation (\ref{xi}) has no solution.
\end{itemize}

It is easy to see that $2<\xi_1<\xi_2$ for all $\theta<\theta_c'$.
This allows to  find all 4 positive solutions to
the equation (\ref{x4}) explicitly, i.e. we have
\begin{equation}\label{x4-7}\begin{array}{ll}
x_{4}={1\over 2}(\xi_2-\sqrt{\xi_2^2-4}), \ \ x_{5}={1\over 2}(\xi_1-\sqrt{\xi_1^2-4}),\\[3mm]
 x_{6}={1\over 2}(\xi_1+\sqrt{\xi_1^2-4}),\ \ x_{7}={1\over 2}(\xi_2+\sqrt{\xi_2^2-4}).
 \end{array}
 \end{equation}
In Fig. \ref{four} the graphs of $x_i$, $i=4,5,6,7$ are shown.
 
Now to find corresponding $y$ we need the following
\begin{lemma}\label{>} For each $x\in \{x_4, x_5,x_6,x_7\}$ and $\theta\leq \theta_c'$ the RHS of
        (\ref{y}) is positive, i.e.
        $$(1-\theta^2)x-\theta^2(x^2+1)>0.$$
\end{lemma}
\begin{proof} We shall use that $x\in \{x_4, x_5,x_6,x_7\}$:
$$(1-\theta^2)x-\theta^2(x^2+1)=x-\theta^2(x^2+x+1)$$ $$=x\left(1-\theta^2\left(x+{1\over x}+1\right)\right)=
x(1-\theta^2(\xi_i+1)), \ \ i=1,2.$$
In case $i=1$ we have
$$1-\theta^2(\xi_1+1)={1\over 2}(1+\theta^2+\sqrt{(\theta-1)(\theta^3+\theta^2+3\theta-1)})$$
which is positive for any $\theta<\theta_c'$.

For $i=2$ we have
$$1-\theta^2(\xi_2+1)={1\over 2}(1+\theta^2-\sqrt{(\theta-1)(\theta^3+\theta^2+3\theta-1)})$$
this number is positive, which can be easily checked for any $0<\theta<\theta_c'$.
\end{proof}
Using this lemma we can define
\begin{equation}\label{y4-7}
y_{i}={1\over \sqrt{\theta}}\sqrt{(1-\theta^2)x_i-\theta^2(x_i^2+1)}, \ \ i=4,5,6,7.
\end{equation}
In Fig.\ref{y1-7} the graphs of all $y_i$, $i=1,2,...,7$ are shown. 
 
 Summarizing we get the full characterization of solutions:
 \begin{pro}\label{p2} The 
 set of solutions to the system (\ref{rs3.2a}), (\ref{rs3.2b}) 
 changes under variations of the parameter $\theta$ is  the following way.  
 There exist $\theta_c (\approx 0.1414)$ and $\theta_c'$ (given by (\ref{tc})) such that
 \begin{itemize}
 \item If $\theta>\theta_c'$ then the system has a unique solution $v_1=(1,y_1)$;
  \item If $\theta=\theta_c'$ then the system has three solutions $v_1=(1,y_1)$, $v_4=(x_4,y_4)$, $v_6=(x_6,y_6)$;
  \item If $\theta_c<\theta<\theta_c'$ then the system has five solutions $v_1=(1,y_1)$, $v_i=(x_i,y_i)$, $i=4,5,6,7$;
  \item If $\theta=\theta_c$ then the system has six solutions $v_1=(1,y_1)$, $v_i=(x_i,y_i)$, $i=3,4,5,6,7$;
  \item If $\theta<\theta_c$ then the system has seven solutions $v_i=(x_i,y_i)$, $i=1,2,3,4,5,6,7$,
   \end{itemize}
 where $y_i$, $i=1,2,3$ are solutions of the equation (\ref{y3}) which can be given explicitly by Cardano's formula, $x_1=x_2=x_3=1$, $x_i$ and $y_i$ for $i=4,5,6,7$ are given by formulas (\ref{x4-7}) and (\ref{y4-7}).  \, (See Fig. \ref{four} and \ref{y1-7} for graphs of these functions.)
 \end{pro}
 
 As an immediate corollary to Propositions \ref{rsp2.1} and \ref{p2} we get
part I of the Theorem \ref{tm}, 
where we denote by $\mu_i$ the TISGM corresponding to $v_i$, $i=1,...,7$.

\section{Tree-indexed Markov chains of TISGMs.}

A tree-indexed Markov chain is defined as follows. Suppose we are 
given a tree with vertices set $V$, and a probability measure $\nu$ 
and a transition matrix $\mathbb{P}=(P_{ij})_{i,j\in \Phi}$ on the
single-site space which is here the 
finite set $\Phi=\{0,1,\dots,m\}$. We can obtain a tree-indexed
Markov chain $X : V \to \Phi$ by choosing $X(x^0)$ according
to $\nu$ and choosing $X(v)$, for each vertex $v\ne x^0$, using
the transition probabilities given the value of its parent,
independently of everything else. See Definition 12.2 in \cite{Ge} for a detailed definition.

We note that a TISGM corresponding to a vector $v=(x,y)\in \mathbb{R}^2$ (which is solution to the system (\ref{rs3.2a}),(\ref{rs3.2b})) is a tree-indexed Markov chain
with states $\{0,1,2\}$ and transition probabilities matrix:
\begin{equation}\label{m} {\mathbb P}=\left(\begin{array}{ccc}
{x^2\over x^2+\theta y^2+\theta^2}&{\theta y^2\over x^2+\theta y^2+\theta^2}&{\theta^2\over x^2+\theta y^2+\theta^2}\\[3mm]
{\theta x^2\over \theta x^2+y^2+\theta}&{y^2\over \theta x^2+y^2+\theta}&{\theta\over \theta x^2+y^2+\theta}\\[3mm]
{\theta^2x^2\over \theta^2x^2+\theta y^2+1}&{\theta y^2\over \theta^2x^2+\theta y^2+1}&{1\over \theta^2x^2+\theta y^2+1}
\end{array}
\right).
\end{equation}
Since $(x,y)$ is a solution to the system (\ref{rs3.2a}),(\ref{rs3.2b}) this matrix can be written in the following form
\begin{equation}\label{m1} {\mathbb P}={1\over Z}\left(\begin{array}{ccc}
x&{\theta y^2\over x}&{\theta^2\over x}\\[2mm]
{\theta x^2\over y}&y&{\theta\over y}\\[2mm]
\theta^2x^2&\theta y^2&1
\end{array}
\right),
\end{equation}
where $Z=\theta^2x^2+\theta y^2+1.$

Simple calculations show that the matrix (\ref{m1}) has three eigenvalues: 1 and
$$\lambda_1(x,y,\theta)={x+y+1-Z-\sqrt{(1+x+y-3Z)^2-4\theta^2Zx^{-1}(1+x^3+y^3)}\over 2Z};
$$
\begin{equation}\label{ev}
\lambda_2(x,y,\theta)=
{x+y+1-Z+\sqrt{(1+x+y-3Z)^2-4\theta^2Zx^{-1}(1+x^3+y^3)}\over 2Z},
\end{equation}
where $\lambda_1$ and $\lambda_2$ are solutions to
\begin{equation}\label{qe}
Zx(1-\lambda)^2+x(1+x+y-3Z)(1-\lambda)+\theta^2(1+x^3+y^3)=0.
\end{equation}

\subsection{Conditions of non-extremality}

It is known (see, e.g., \cite{Ge}) that for all
$\beta >0$, the Gibbs measures  form a non-empty
convex compact set in the space of probability measures.
Extreme measures, i.e., extreme points of this set are
associated with pure phases.
Furthermore, any Gibbs measure is an integral of
extreme ones (the extreme decomposition).
Thus extreme points are important to describe the convex set of all Gibbs measures. 
In this subsection we are going to find the regions of the parameter $\theta$ where the 
TISGMs $\mu_i$, $i=1,\dots,7$ are not extreme in the set of all Gibbs measures (including 
the non-translation invariant ones). 

It is known that a sufficient (Kesten-Stigum) condition for non-extremality of a Gibbs measure $\mu$ corresponding
to the matrix ${\mathbb P}$ on a Cayley tree of order $k\geq 1$ is that
$k\lambda^2_2>1$, where $\lambda_2$ is the second largest (in absolute value) eigenvalue of ${\mathbb P}$ \cite{Ke}.
We are going to use this condition for TISGMs $\mu_i$, $i=1,\dots,7$.
We have all solutions of the system (\ref{rs3.2a}),(\ref{rs3.2b}) and the eigenvalues of the matrix ${\mathbb P}$ in the explicit form. But these quantities have very complicated long form which are functions of the one real variable $\theta$ only, and will use a computer for a numerical 
investigation of the relevant properties of the function. 

Let us denote
$$\lambda_{\max,i}(\theta)=\max\{|\lambda_1(x_i,y_i,\theta)|, |\lambda_2(x_i,y_i,\theta)|\}, \ \ i=1,\dots,7.$$

Using a computer one can obtain the graphs of the discriminant of the equation (\ref{qe}) and note 
that
$\lambda_1(x_i,y_i,\theta)$ and $\lambda_2(x_i,y_i,\theta)$ are
real for any $i=1,\dots,7.$
Moreover, we have
$$\lambda_{\max,i}(\theta)=\left\{\begin{array}{llll}
|\lambda_2(x_1,y_1,\theta)|, \ \ \mbox{if} \ \ i=1, \ \ \theta<1\\[2mm]
|\lambda_1(x_1,y_1,\theta)|, \ \ \mbox{if} \ \ i=1, \ \ \theta>1\\[2mm]
|\lambda_2(x_i,y_i,\theta)|, \ \ \mbox{if} \ \ i=2,3,4,5,6,7.
\end{array}
\right.$$
Denote
$$\eta_i(\theta)=2\lambda^2_{\max,i}(\theta)-1,\ \ i=1,\dots,7.$$

In Fig. \ref{ny1ny2} - Fig.\ref{ny7} the graphs of the functions $\eta_i$ are shown. Note that 
 these are only functions of $\theta$ and do not have any additional parameter. From the 
 graphs one can see the regions of $\theta$ for which the corresponding function is positive.    

Thus we obtained 
the following proposition (which gives the results of the part II of Theorem \ref{tm} 
concerning to the non-extremality). Note that the parameter values we present might not 
be optimal, as it is not clear whether the
 Kesten-Stigum condition is optimal in our model, but there are the optimal values 
 which are provided by the Kesten-Stigum condition. 

\begin{pro}\label{tne}
\begin{itemize}

\item[1)] There exists a value $\bar{\bar{\theta}} \ \ (\approx 2.8765)$ such that the measure $\mu_1$ is non-extreme for any $\theta>\bar{\bar{\theta}}$.

\item[2)] For TISGMs $\mu_i$, $i=2,3$ the Kesten-Stigum condition is always satisfied, i.e. these measures are non-extreme for all values of $\theta$ for which they exist.

\item[3)] There exists  a value 
$\theta^*\ \ (\approx 0.171719)$ such that the measures $\mu_5$ and $\mu_6$ 
are non-extreme for any $\theta<\theta^*$.

\item[4)] For TISGMs $\mu_4$ and $\mu_7$ the Kesten-Stigum condition is never hold.
\end{itemize}
\end{pro}

 
 \subsection{Conditions for extremality}
 
 In this subsection we are going to find sufficient conditions for extremality (or non-reconstructability in information-theoretic language \cite{MSW},\cite{Mos2},\cite{Mos},\cite{Sly11}) of TISGMs for the 3-state SOS model,
 depending on coupling strength parameterized by $\theta$ and the boundary law.  
 By the above-mentioned non-extremality conditions we know that $\mu_i$, $i=2,3$ are not-extreme, so in this subsection we shall consider the remaining TISGMs: $\mu_i$, $i=1,4,5,6,7$. 
 
 We will prove the following proposition 
 (which gives the results of part II of Theorem \ref{tm} concerning sufficient conditions 
 for extremality.) 
 
 \begin{pro}\label{te} The following assertions hold. 
 \begin{itemize}
 \item[(a)] The exists $\bar{\theta}$ \ \ $(\approx 2.656)$ such that the measure $\mu_1$ is extreme for any $\theta<\bar{\theta}$ (see Fig.\ref{ney1}).
 
 \item[(b)] There exists $\theta^{**}<\theta_c'$ such that the measures $\mu_5$ and $\mu_6$ are extreme for any $\theta>\theta^{**}$.
 
\item[(c)] The measures $\mu_4$ and $\mu_7$ are extreme as soon as they exist (see Fig. \ref{ney6y7}).
 \end{itemize}
 \end{pro}
 The proof of this proposition follows after the following subsection.


 \subsection{Reconstruction insolvability on trees: extremality of TISGM}
 
 To prove Proposition \ref{te} we will use a result of \cite{MSW} to establish a bound for reconstruction impossibility corresponding to the matrix (channel) of a solution $v_i$, $i=1,4,5,6,7$.
 
 Let us first give some necessary definitions from \cite{MSW}. Considering finite complete subtrees $\mathcal T$ that are initial points of Cayley tree $\Gamma^k$, i.e. share the same root; if $\mathcal T$ has depth $d$ (i.e. the vertices of $\mathcal T$ are within distance $\leq d$ from the root)
 then it has $(k^{d+1}-1)/(k-1)$ vertices, and its boundary $\partial \mathcal T$ consists the neighbors (in $\Gamma^k\setminus \mathcal T$) of its vertices, i.e., $|\partial \mathcal T|=k^{d+1}.$  We identify subgraphs of $\mathcal T$ with their vertex sets and write $E(A)$ for the edges within a subset $A$ and $\partial A$ for the boundary of $A$, i.e., the neighbors of $A$ in $(\mathcal T\cup \partial\mathcal T)\setminus A)$.
 
 In \cite{MSW} the key ingredients are two quantities,
 $\kappa$ and $\gamma$, which bound the rates of
 percolation of disagreement down and up the tree, respectively. Both are properties of
 the collection of Gibbs measures $\{\mu^\tau_{{\mathcal T}}\}$, where the boundary condition $\tau$
 is fixed and $\mathcal T$ ranges over all initial finite complete subtrees of $\Gamma^k$.
  For a given subtree $\mathcal T$ of $\Gamma^k$ and a vertex $x\in\mathcal T$, we write $\mathcal T_x$ for
 the (maximal) subtree of $\mathcal T$ rooted at $x$. When $x$ is not the root of $\mathcal T$, let $\mu_{\mathcal T_x}^s$
 denote the (finite-volume) Gibbs measure in which the
 parent of $x$ has its spin fixed to $s$ and the configuration on the bottom boundary  of ${\mathcal T}_x$
 (i.e., on $\partial {\mathcal T}_x\setminus \{\mbox{parent\ \ of}\ \ x\}$) is
 specified by $\tau$.
 
  For two measures $\mu_1$ and $\mu_2$ on $\Omega$, $\|\mu_1-\mu_2\|_x$ denotes the variation distance between the projections of $\mu_1$ and $\mu_2$ onto the spin at $x$, i.e.,
 $$\|\mu_1-\mu_2\|_x={1\over 2}\sum_{i=0}^2|\mu_1(\sigma(x)=i)-\mu_2(\sigma(x)=i)|.$$
 Let $\eta^{x,s}$ be the
 configuration $\eta$ with the spin at $x$ set to $s$.
 
 Following \cite{MSW} define
 $$\kappa\equiv \kappa(\mu)=\sup_{x\in\Gamma^k}\max_{x,s,s'}\|\mu^s_{{\mathcal T}_x}-\mu^{s'}_{{\mathcal T}_x}\|_x;$$
 $$\gamma\equiv\gamma(\mu)=\sup_{A\subset \Gamma^k}\max\|\mu^{\eta^{y,s}}_A-\mu^{\eta^{y,s'}}_A\|_x,$$
 where the maximum is taken over all boundary conditions $\eta$, all sites $y\in \partial A$, all neighbors $x\in A$ of $y$, and all spins $s, s'\in \{0,1,2\}$.
 
 As the main ingredient we apply \cite[Theorem 9.3]{MSW}, which is 
 
 \begin{thm} For an arbitrary (ergodic\footnote{Ergodic means irreducible and aperiodic Markov chain. Therefore has a unique stationary distribution $\pi=(\pi_1,\dots,\pi_q)$ with $\pi_i>0$ for all $i$.} and permissive\footnote{Permissive means that for arbitrary finite $A$ and boundary condition outside $A$ being $\eta$ the conditioned Gibbs measure on $A$, corresponding to the channel is positive for at least one configuration.}) channel ${\mathbb P}=(P_{ij})_{i,j=1}^q$
 on a tree, the reconstruction of the corresponding tree-indexed Markov chain 
 is impossible if $k\kappa\gamma<1$. 
 \end{thm}
 It is easy to see that the channel ${\mathbb P}$ corresponding to a TISGM of the SOS model is ergodic and permissive. Thus the criterion of {\it extremality} of a TISGM is $k\kappa\gamma<1$.
 
 Note that $\kappa$ has the particularly simple form (see \cite{MSW})
 \begin{equation}\label{ka}
 \kappa={1\over 2}\max_{i,j}\sum_l|P_{il}-P_{jl}|
 \end{equation}
 and $\gamma$ is a constant which does not have a clean general formula, but can be estimated in specific models (as Ising, Hard-Core etc.). For example, if ${\mathbb P}$ is the symmetric channel of the Potts model then $\gamma\leq {\theta-1\over\theta+1}$  \cite[Theorem 8.1]{MSW}.
 
 \begin{rk}
 Since each TISGM $\mu$ corresponds to a solution $(x,y)$ of the system of equations (\ref{rs3.2a}), (\ref{rs3.2b}) we can write $\gamma(\mu)=\gamma(x,y)$ and $\kappa(\mu)=\kappa(x,y)$. 
 \end{rk}
 
 \subsubsection{The estimation of $\gamma$ for the SOS model.}
 
 After generalities of the approach of Martinelli, Sinclair, Weitz we are now 
 ready to start to technical work to estimate the constant $\gamma(x_i,y_i)$ 
 depending on the boundary law labeled by $i$ from above.

 Consider the case $v_i$, $i=1,4,5,6,7$.
 
  \begin{pro}\label{pd} Recall the matrix ${\mathbb P}$, given by (\ref{m}), and denote by
 $\mu=\mu(\theta)$ the corresponding Gibbs measure. Then, for any subset $A\subset {\mathcal T}$, (where $\mathcal T$ is initial complete subtree of $\Gamma^k$)
  any boundary configuration $\eta$, any pair of spins
 $(s_1, s_2)$, any site $y\in \partial A$, and any neighbor $x\in A$ of $y$, we have
  $$\|\mu^{\eta^{y,s_1}}_A-\mu^{\eta^{y,s_2}}_A\|_x=
  \max\{|p^{0}(0)-p^{2}(0)|, |p^{2}(2)-p^{0}(2)|\},$$
 where  $p^{t}(s)=\mu^{\eta^{y,t}}_A(\sigma(x)=s)$.
 \end{pro}
 \begin{proof} Denote $p_s=\mu^{\eta^{y,\,free}}_A(\sigma(x)=s)$, $s=0,1,2$.   By definition of the matrix ${\mathbb P}$ we have
$$
 p^0(0)={x^2p_0\over x^2p_0+\theta y^2p_1+\theta^2 p_2}, \ \
  p^0(1)={\theta y^2p_1\over x^2p_0+\theta y^2p_1+\theta^2 p_2},\ \
   p^0(2)={\theta^2p_2\over x^2p_0+\theta y^2p_1+\theta^2 p_2};$$
    \begin{equation}\label{ppp}
  p^1(0)={\theta x^2p_0\over \theta x^2p_0+ y^2p_1+\theta p_2}, \ \
     p^1(1)={y^2p_1\over \theta x^2p_0+ y^2p_1+\theta p_2},\ \
      p^1(2)={\theta p_2\over \theta x^2p_0+ y^2p_1+\theta p_2};
    \end{equation}   
    $$
       p^2(0)={\theta^2x^2p_0\over \theta^2 x^2p_0+\theta y^2p_1+ p_2}, \ \
        p^2(1)={\theta y^2p_1\over \theta^2 x^2p_0+\theta y^2p_1+ p_2},\ \
         p^2(2)={p_2\over \theta^2 x^2p_0+\theta y^2p_1+p_2}.
        $$

 The  proposition follows from the following Lemma \ref{lem} and Lemma \ref{lem1}.\end{proof}
 
 \begin{lemma}\label{lem} \begin{itemize}
 \item[(i)] If $\theta<1$ then 
 \begin{itemize}
 \item[a)]  $p^{0}(0)\geq p^{1}(0)\geq p^2(0)$;
 
 \item[b)]  $p^{1}(1)\geq p^{0}(1)\geq p^2(1)$ 
 \ \ \mbox{if} \ \ $p_2\geq x^2 p_0$ and 
 $p^{1}(1)\geq p^{2}(1)\geq p^0(1)$ 
  \ \ \mbox{if} \ \ $p_2\leq x^2 p_0$;
 
 \item[c)]  $p^{2}(2)\geq p^{1}(2)\geq p^0(2)$.
 \end{itemize}
 
  \item[(ii)] If $\theta>1$ then 
  \begin{itemize}
  \item[a)]  $p^{0}(0)\leq p^{1}(0)\leq p^2(0)$;
  
  \item[b)]  $p^{1}(1)\leq p^{0}(1)\leq p^2(1)$ 
  \ \ \mbox{if} \ \ $p_2\geq x^2 p_0$ and 
  $p^{1}(1)\leq p^{2}(1)\leq p^0(1)$ 
   \ \ \mbox{if} \ \ $p_2\leq x^2 p_0$;
  
  \item[c)]  $p^{2}(2)\leq p^{1}(2)\leq p^0(2)$.
  \end{itemize}
 \end{itemize}
 \end{lemma}
 \begin{proof}  We shall prove two inequalities (all others are very similar): 
 by the formula (\ref{ppp}) we get
  $$p^{0}(0)-p^{1}(0)={x^2p_0(1-\theta^2)(y^2p_1+\theta p_2)\over \left(x^2p_0+\theta y^2p_1+\theta^2 p_2\right)\left(\theta x^2p_0+ y^2p_1+\theta p_2\right)},$$
  which is positive iff $\theta<1$.
   $$p^{0}(1)-p^{2}(1)={\theta y^2p_1(1-\theta^2)(p_2-x^2 p_0)\over \left(x^2p_0+\theta y^2p_1+\theta^2 p_2\right)\left(\theta^2 x^2p_0+\theta y^2p_1+p_2\right)},$$
   which is non-negative if $\theta<1$ and $p_2\geq x^2 p_0$ or $\theta>1$ and $p_2\leq x^2 p_0$.
   \end{proof}
 From Lemma \ref{lem} we obtain the following lemma. 
 \begin{lemma}\label{lem1} We have 
 $$\max_{i,j,k}\left\{|p^i(k)-p^j(k)|\right\}=\max\{|p^0(0)-p^2(0)|, |p^2(2)-p^0(2)|\}.$$
 \end{lemma}
 \begin{proof} Consider the case $\theta<1$ the other case is similar. We have 
 $$ p^0(0)-p^1(0)=p^1(1)-p^0(1) + p^1(2)-p^0(2). $$
 By Lemma \ref{lem} we have $p^1(1)-p^0(1)\geq 0$ and $p^1(2)-p^0(2)\geq 0$. Thus  
 $$ p^0(0)-p^1(0)\geq p^1(1)-p^0(1),$$
  $$ p^0(0)-p^1(0)\geq p^1(2)-p^0(2).$$ Again using Lemma \ref{lem}
  we get 
  $$ p^0(0)-p^2(0)\geq p^0(0)-p^1(0)\geq p^1(1)-p^0(1),$$
   $$ p^2(2)-p^0(2)\geq p^2(2)-p^1(2)\geq p^1(1)-p^2(1),$$
   and
   $$ p^0(0)-p^2(0)\geq p^1(0)-p^2(0).$$
   Now,  from the equality 
    $$ p^0(1)-p^2(1)= p^2(2)-p^0(2)-(p^0(0)-p^1(0))$$
    it follows that 
    $$ |p^0(1)-p^2(1)|\leq \max\{|p^2(2)-p^0(2)|,|p^0(0)-p^1(0)|\}.$$
 \end{proof}
 Let $p=(p_0,p_1,p_2)$ be a probability distribution on $\{0,1,2\}$. For $t=p_0$ and $u=p_2$, $0\leq t+u\leq 1$ we define the following functions
 $$f(t,u,\theta)=p^0(0)-p^2(0)={x^2t\over (x^2-\theta y^2)t+\theta(\theta-y^2)u+\theta y^2}-{x^2\theta^2t\over \theta(\theta x^2-y^2)t+(1-\theta y^2)u+\theta y^2};$$
 $$g(t,u,\theta)=p^2(2)-p^0(2)={u\over \theta(\theta x^2-y^2)t+(1-\theta y^2)u+\theta y^2}-{\theta^2 u\over (x^2-\theta y^2)t+\theta(\theta- y^2)u+\theta y^2}.$$
 
 \begin{lemma}\label{lK} We have
 $$|f(t,u,\theta)|\leq {|1-\theta^2|\over 1+\theta^2} \ \ \mbox{and} \ \ |g(t,u,\theta)|\leq {|1-\theta^2|\over 1+\theta^2}.$$
 \end{lemma}
 \begin{proof} We shall consider the case $\theta<1$, because the case $\theta>1$ is similar, moreover the upper bound which we want to prove is invariant under replacement of  $\theta$ by $1/\theta$.  To find the maximal value of the function we have to solve the following system 
 \begin{equation}\label{hu}
 f_u'(t,u,\theta)={\theta x^2t(y^2-\theta)\over ((x^2-\theta y^2)t+\theta(\theta-y^2)u+\theta y^2)^2}+{\theta^2 x^2t(1-\theta y^2)\over (\theta(\theta x^2-y^2)t+(1-\theta y^2)u+\theta y^2)^2}=0
 \end{equation}
  \begin{equation}\label{ht}
  f_t'(t,u,\theta)={\theta x^2(y^2+(\theta-y^2)u)\over ((x^2-\theta y^2)t+\theta(\theta-y^2)u+\theta y^2)^2}-{\theta^2 x^2((1-\theta y^2)u+\theta y^2)\over (\theta(\theta x^2-y^2)t+(1-\theta y^2)u+\theta y^2)^2}=0.
  \end{equation}
  From (\ref{hu}) one has either $t=0$,  or if $t\ne 0$ 
  we note that if $y^2=1/\theta$ then $y^2=\theta$, i.e. $\theta=1$. 
  So we can assume $y^2\ne 1/\theta$.  Then from (\ref{hu}) we get (for $t\ne 0$) that
 $$\left({(x^2-\theta y^2)t+\theta(\theta-y^2)u+\theta y^2\over \theta(\theta x^2-y^2)t+(1-\theta y^2)u+\theta y^2}\right)^2=
 {\theta-y^2\over \theta(1-\theta y^2)}
 $$ 
 and from (\ref{ht}) we get 
 $$\left({(x^2-\theta y^2)t+\theta(\theta-y^2)u+\theta y^2\over \theta(\theta x^2-y^2)t+(1-\theta y^2)u+\theta y^2}\right)^2=
  {(\theta-y^2)u+y^2\over \theta((1-\theta y^2)u+\theta y^2)}.
  $$ 
  Thus we should have   
 $${\theta-y^2\over 1-\theta y^2}= {(\theta-y^2)u+y^2\over (1-\theta y^2)u+\theta y^2},$$
 which is possible only iff $\theta=1$. So it remains only the case $t=0$ which gives a minimum ($=0$) of the function $f$. Hence the maximal value of $f$  is reached on the boundary of the set 
 $\{(t,u)\in [0,1]^2: t+u\leq 1\}$.  We discuss the three line segments 
 of the boundary separately: 
 
 {\it Case: $t=0$.} In this case it was already 
 mentioned above that the function has a minimum which is equal to zero.
 
{\it Case: $u=0$.} In this case simple calculations show that 
$$\max  f(t,0,\theta) =f\left({y^2\over x^2+y^2},0,\theta\right)={1-\theta\over 1+\theta}.$$

{\it Case: $t+u=1$.} In this case simple calculations show that 
$$\max  f(t,1-t,\theta) =f\left({1\over 1+x^2},{x^2\over 1+x^2},\theta\right)={1-\theta^2\over 1+\theta^2}.$$

Note that $ {1-\theta \over 1+\theta}\leq  {1-\theta^2\over 1+\theta^2}$ this completes the proof for $f$.  For $g$ the proof is very similar.
 \end{proof}
 
The following proposition gives a bound for $\gamma$.
 
 \begin{pro} \label{peg} Independently on the possible 
 values of $(x,y)$ (i.e. the solutions to the system (\ref{rs3.2a}),(\ref{rs3.2b})) we have  
 \begin{equation}
 \label{gam} 
 \gamma(x,y)\leq  {|1-\theta^2|\over 1+\theta^2}.
 \end{equation}
 \end{pro}
 \begin{proof}
 This is a corollary of above-mentioned lemmas.
  \end{proof}
 
\subsubsection{Computation of $\kappa$.} Now we shall compute the constant $\kappa$.

  Using (\ref{ka}) and (\ref{m1}) we get
  $$
 \kappa(x,y)=
 {1\over 2}\max_{i,j}\sum_{l=0}^2|P_{il}-P_{jl}|={1\over 2Z}\max\left\{{x^2|y-\theta x|+(y^2+\theta)|x-\theta y|\over xy}\right.,$$
  \begin{equation}\label{kab}
  \left.
  {x^2|1-\theta^2 x|+\theta y^2|1-x|+|\theta^2-x|\over x},\ \
  {(\theta x^2+y^2)|1-\theta y|+|\theta- y|\over y}
 \right\}, \end{equation}
 where $Z=\theta^2x^2+\theta y^2+1$.
 
 We shall compute $\kappa(x,y)$ for $(x,y)\in \{(1,y_1),(x_4,y_4),(x_5,y_5),(x_6,y_6),(x_7,y_7)\}$. 
 
 For the solution $(1,y_1)$ we have 
 $$y_1=y_1(\theta)=\left\{\begin{array}{lll}
 \in (1,{1\over \theta}), \ \ \mbox{if} \ \ \theta<1;\\[3mm]
 1, \ \ \mbox{if}  \ \ \theta=1;\\[3mm]
 \in({1\over \theta},1), \ \ \mbox{if} \ \ \theta>1.
 \end{array}\right.$$
 Using these relations and the fact that $y_1$ is a solution of (\ref{y3}) from (\ref{kab}) we get
 \begin{equation}\label{k1}
 \kappa(1,y_1)={|1-\theta^2|\over 1+\theta^2+\theta y_1^2}.
  \end{equation}
Now we shall compute $\kappa$ for $(x_i,y_i)$, $i=4,5,6,7$. Recall that all of them exist only when $\theta\leq \theta_c'$. Moreover, $x_i<1$ if $i=4,5$ and $x_i>1$ if $i=6,7$. So for $\theta<1$ from the system (\ref{rs3.2a}),(\ref{rs3.2b}) we get 
the following inequalities
$$y-\theta x={(1-\theta^2)(y^2+\theta)\over Z}>0, \ \ x-\theta y={x^2(1-\theta^2)\over Z}>0, $$ $$1-\theta^2x={(1-\theta^2)(\theta y^2+\theta^2+1)\over Z}>0, \ \ x-\theta^2={(1-\theta^2)((\theta^2+1)x^2+\theta y^2)\over Z}>0,$$
$$1-\theta y={1-\theta^2\over Z}>0, \ \ y-\theta={(1-\theta^2)(\theta x^2+y^2)\over Z}>0.$$
Using these inequalities, we obtain from the equality (\ref{kab}) that 
$$\kappa(x,y)={1-\theta^2\over Z^2xy}\max\{x^2(y^2+\theta), \ \ y(\theta y^2+(\theta^2+1)x^2),\ \ x(\theta x^2+y^2)\}.$$
It is easy to see that 
$$\max\{x_i^2(y_i^2+\theta), \ \ y_i(\theta y_i^2+(\theta^2+1)x_i^2),\ \ x_i(\theta x_i^2+y_i^2)\}=$$
$$\left\{\begin{array}{lll}
y_i(\theta y_i^2+(\theta^2+1)x^2_i), \ \ \mbox{if} \ \ i=4\\[2mm]
x_i(\theta x_i^2+y_i^2),  \ \ \mbox{if} \ \ i=5\\[2mm]
x_i^2(y_i^2+\theta), \ \ \mbox{if} \ \ i=6,7.
\end{array}\right.$$
Hence we get 
 \begin{equation}\label{k2}
 \kappa(x_i,y_i)=\left\{\begin{array}{lll}
  {(1-\theta^2)(\theta y_4^2+(\theta^2+1)x^2_4)\over x_4(1+\theta^2x_4^2+\theta y_4^2)^2}, \ \ \mbox{if} \ \ i=4;\\[3mm]
   {(1-\theta^2)(\theta x_5^2+y^2_5)\over y_5(1+\theta^2x_5^2+\theta y_5^2)^2}, \ \ \mbox{if} \ \ i=5;\\[3mm]
 {x_i(1-\theta^2)(y_i^2+\theta)\over y_i(1+\theta^2x_i^2+\theta y_i^2)^2}, \ \ i=6,7.
 \end{array}\right.
  \end{equation}
\subsection{Proof of Proposition \ref{te}}

{\it Proof of} (a): To check extremality of TISGM $\mu_1$ we should check 
$2\kappa\gamma<1$. Using the above mentioned bound of $\gamma$ and formula (\ref{k1}) we will check   
$$2\kappa(1,y_1(\theta))\gamma(1,y_1(\theta))\leq  {2(1-\theta)^2\over (1+\theta^2)(1+\theta^2+\theta y_1^2(\theta))}<1.$$
Denote
$$U_1(\theta)={2(1-\theta)^2\over (1+\theta^2)(1+\theta^2+\theta y_1^2(\theta))}-1.$$ 

The function $U_1(\theta)$ only depends on $\theta$ and has no additional parameters.  
From its graph one can see the region of $\theta$ where the function is negative.
Thus looking  on the graph of $U_1(\theta)$ (see Fig.\ref{ey1}) completes the arguments for 
part a).
 
{\it Proof of} (b) and (c): Consider the following functions

$$U_4(\theta)= {2(1-\theta^2)^2(\theta y_4^2+(\theta^2+1)x^2_4)\over x_4(1+\theta^2)(1+\theta^2x_4^2+\theta y_4^2)^2}-1,$$
$$U_5(\theta)= {2(1-\theta^2)^2(\theta x_5^2+y^2_5)\over y_5(1+\theta^2)(1+\theta^2x_5^2+\theta y_5^2)^2}-1,$$

 $$U_i(\theta)={2x_i(1-\theta^2)^2(y_i^2+\theta)\over y_i(1+\theta^2)(1+\theta^2x_i^2+\theta y_i^2)^2}-1, \ i=6,7.$$
 By above mentioned formula (\ref{k2}) and the bound of $\gamma$ we get  
 $$2\kappa(x_i,y_i)\gamma(x_i,y_i)-1\leq U_i(\theta).$$
 Now the proof of part (b) follows from the behavior of the 
 graph of $U_i(\theta)$, $i=5,6$. The proof of  
 part (c) follows from the behavior of the 
 graph of $U_i(\theta)$ for $i=4,7$ (see Fig.\ref{ey4} and Fig. \ref{ey5}).
 
This finishes our proof of  Theorem \ref{tm} which 
summarized the results of Propositions \ref{p2}, \ref{tne}, and \ref{te}.  
 \section*{ Acknowledgements}
 
 U.A. Rozikov thanks the  DFG
 Sonderforschungsbereich SFB $|$ TR12-Symmetries and Universality in Mesoscopic Systems
 and the Ruhr-University Bochum (Germany) for financial support and hospitality.

\section{Figures}      
  \begin{figure}
    \begin{center}
   \includegraphics[width=9cm]{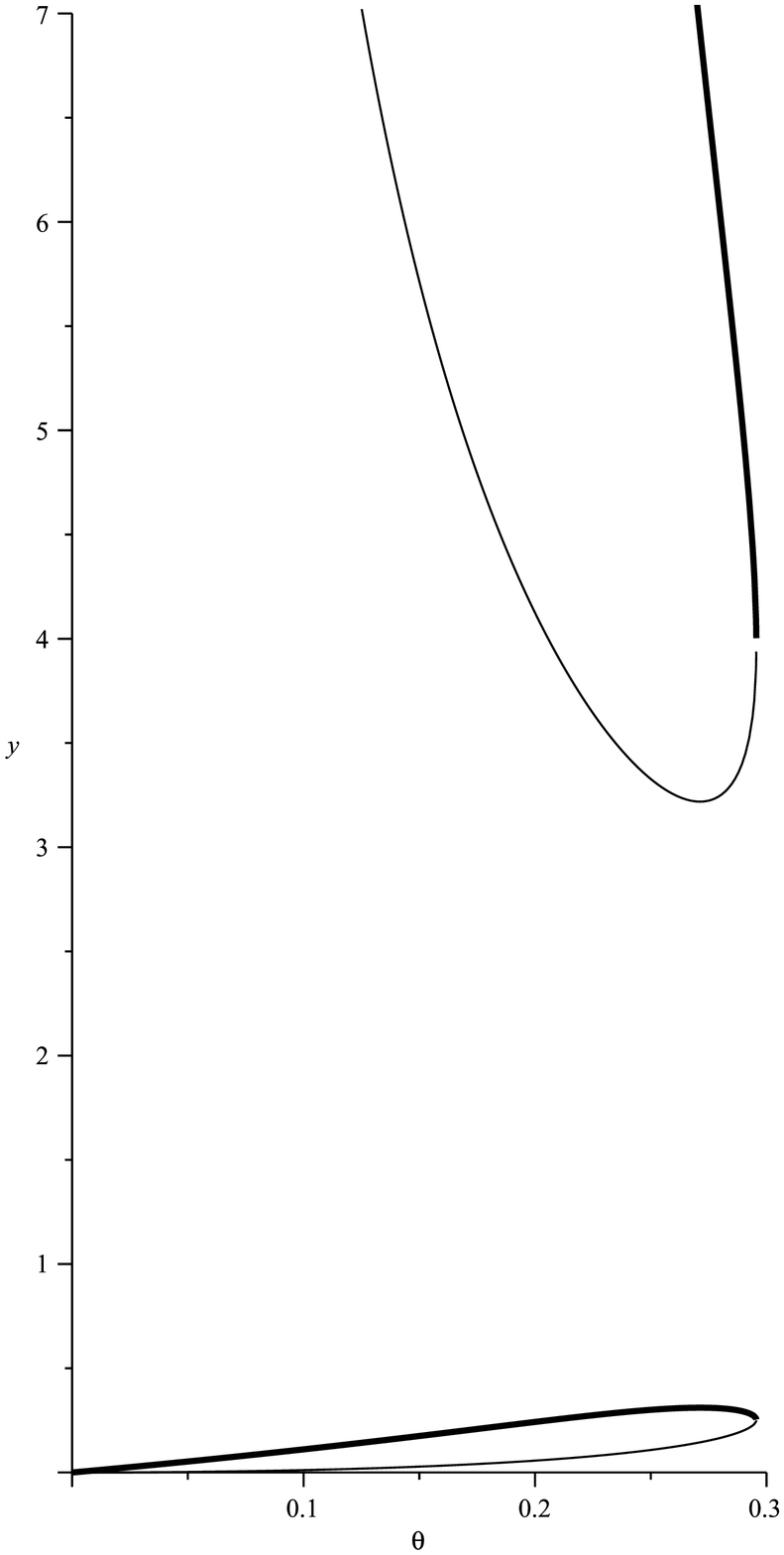}
   \end{center}
   \caption{ The graphs of functions $x_i=x_i(\theta)$, $i=4,5,6,7$. Upper thin curve is $x_6$ and lower thin curve is $x_4$. Upper bold curve is $x_7$ and lower bold curve is $x_5$.}\label{four}
   \end{figure}
    \begin{figure} 
    \begin{center}
    \includegraphics[width=12cm]{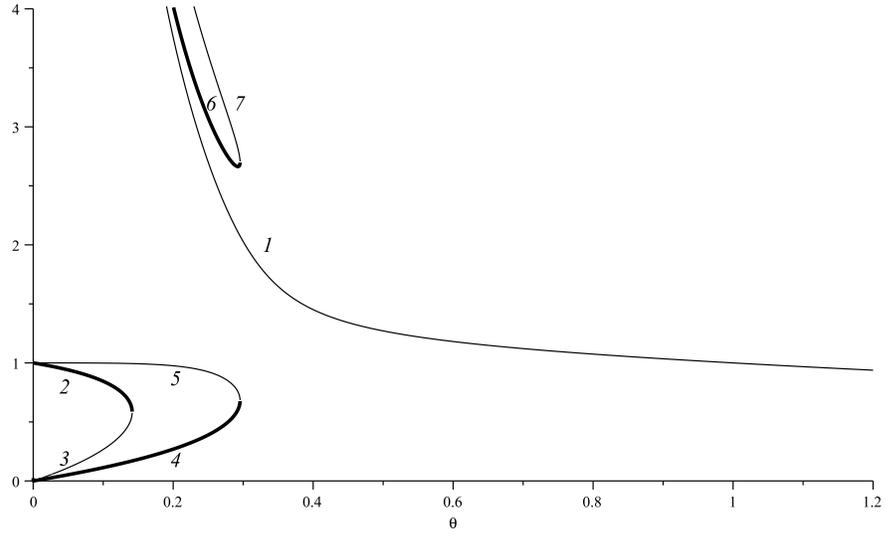}
    \end{center}
    \caption{
     The graphs of functions $y_i=y_i(\theta)$, $i=1,2,...,7$. }\label{y1-7}
    \end{figure}
  \begin{figure}
  \begin{center}
  \includegraphics[width=5.6cm]{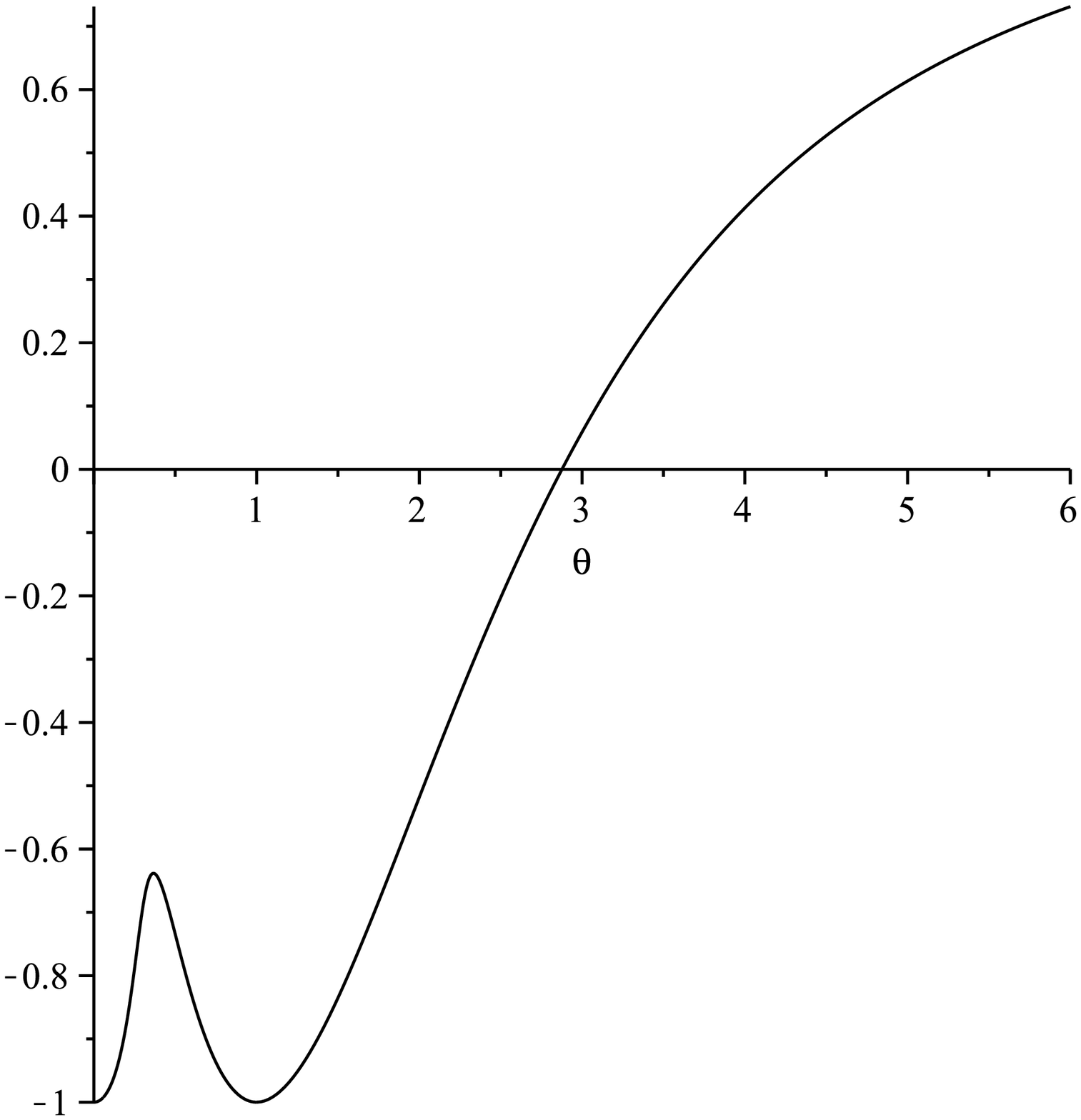} \ \ \ \ \includegraphics[width=5.6cm]{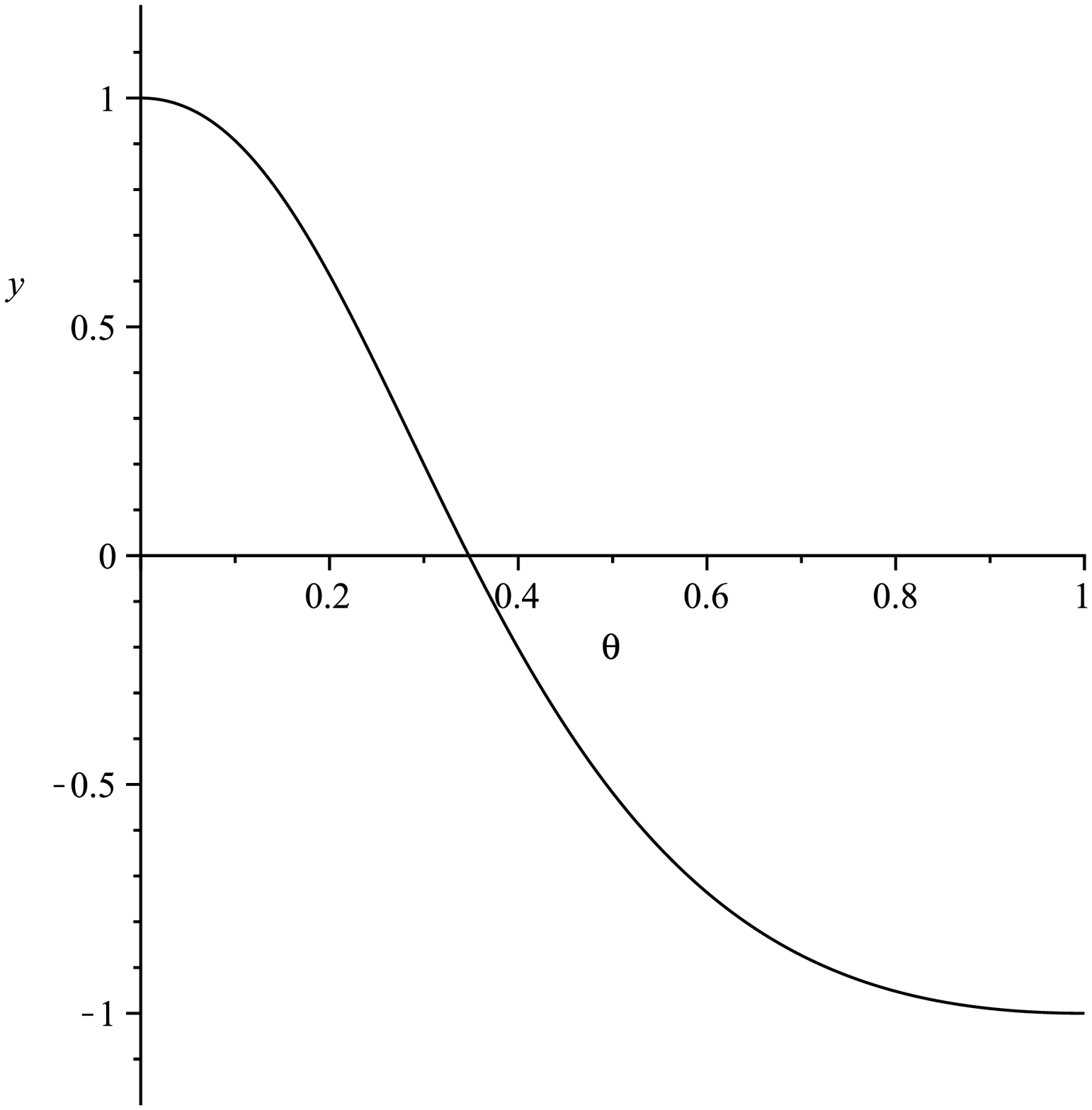}\ \ \ \ 
   \end{center}
  \caption{The graphs of functions $\eta_1(\theta)$ (left). The graph of  $\eta_1(1/\theta)$ (right) for $\theta\in (0,1)$ which shows that $\eta_1(\theta)$ is an increasing function in $(1,+\infty)$   .}\label{ny1ny2}
   \end{figure}
   \begin{figure}
   \begin{center}
   \includegraphics[width=6.5cm]{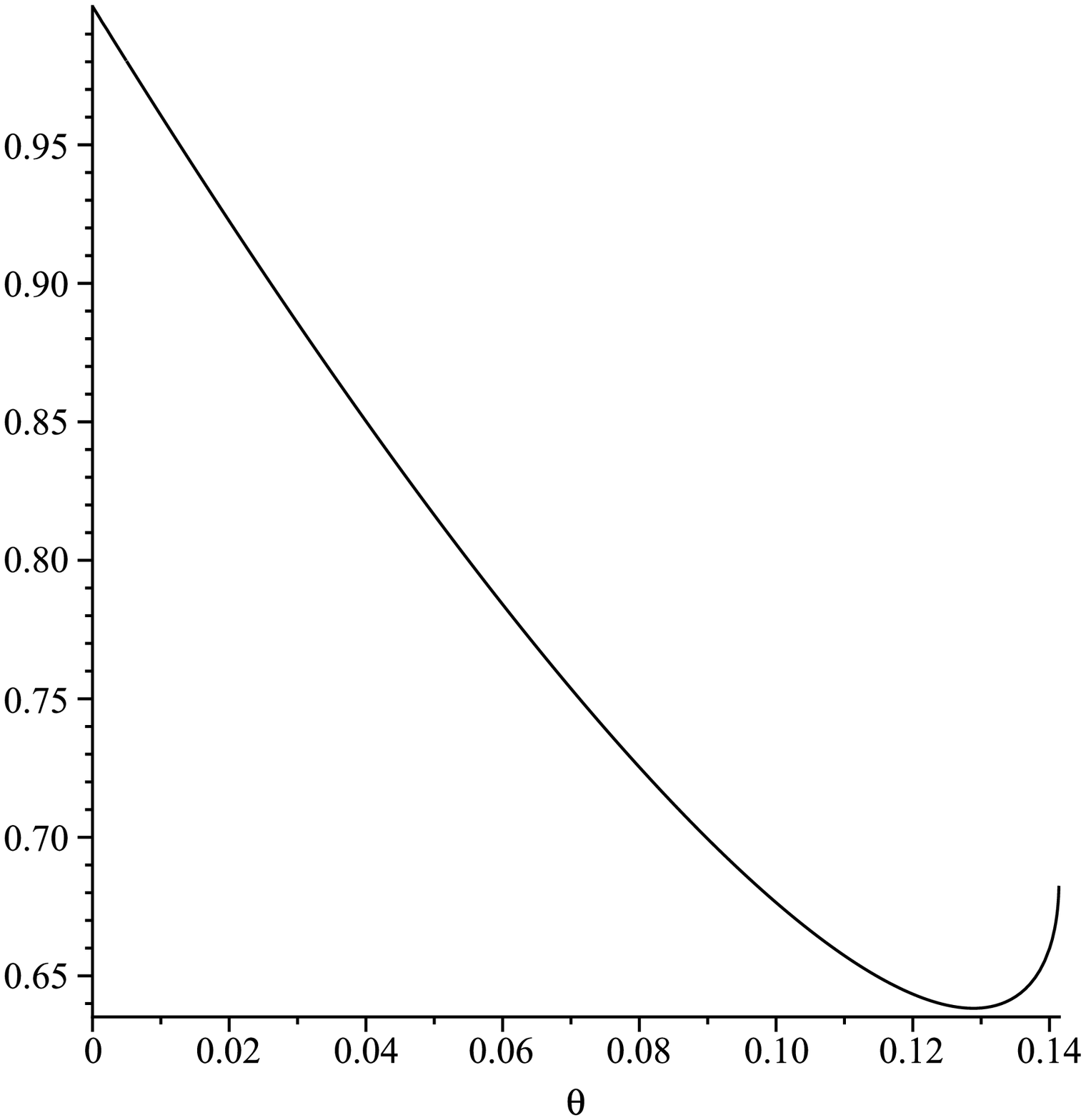}\ \ \ \
  \includegraphics[width=6.5cm]{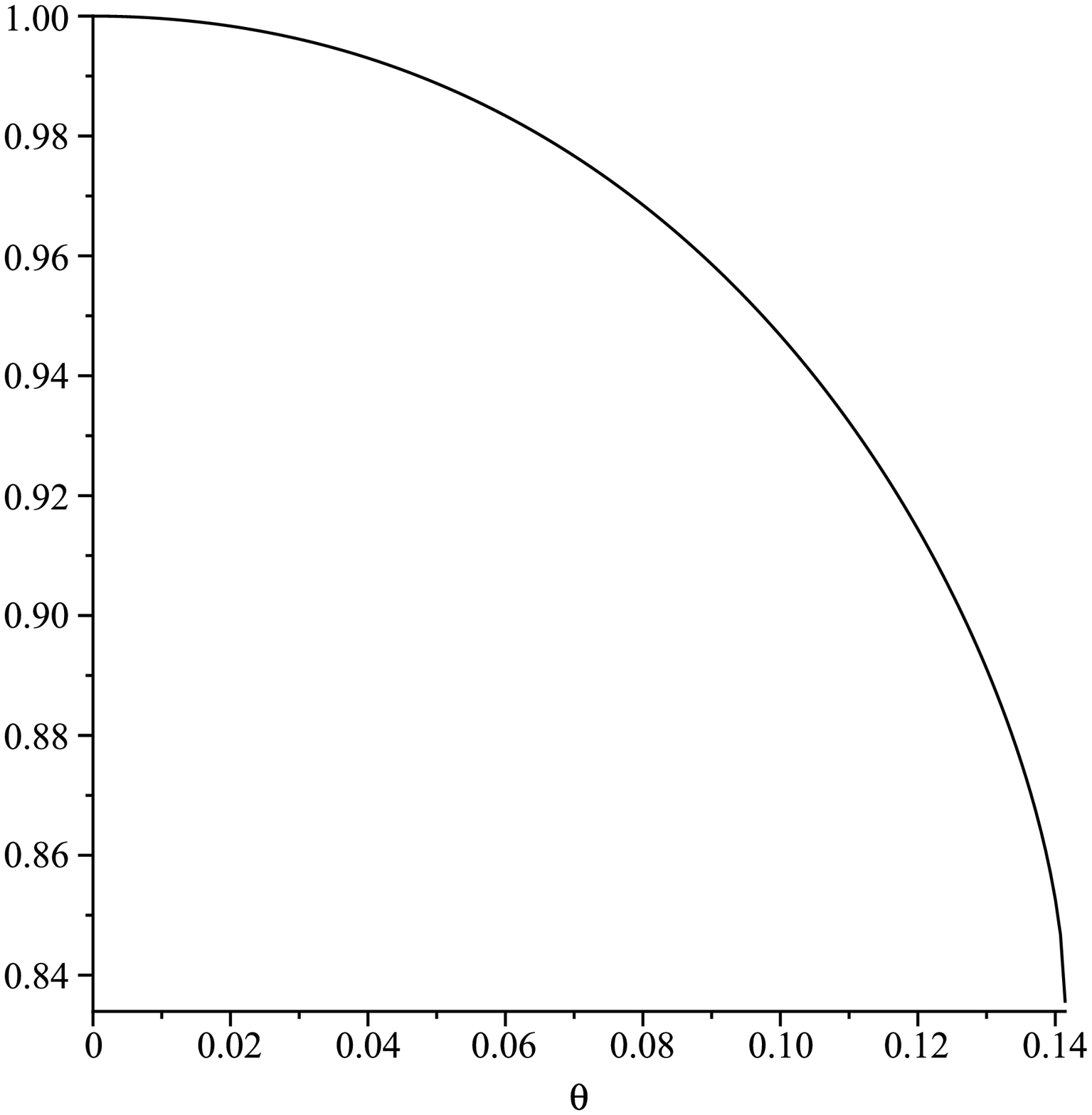} 
  \end{center}
  \caption{The graphs of functions $\eta_2(\theta)$ (left) and $\eta_3(\theta) $ (right).}\label{ny3ny4}
    \end{figure}
   \begin{figure}
   \begin{center}
  \includegraphics[width=6.5cm]{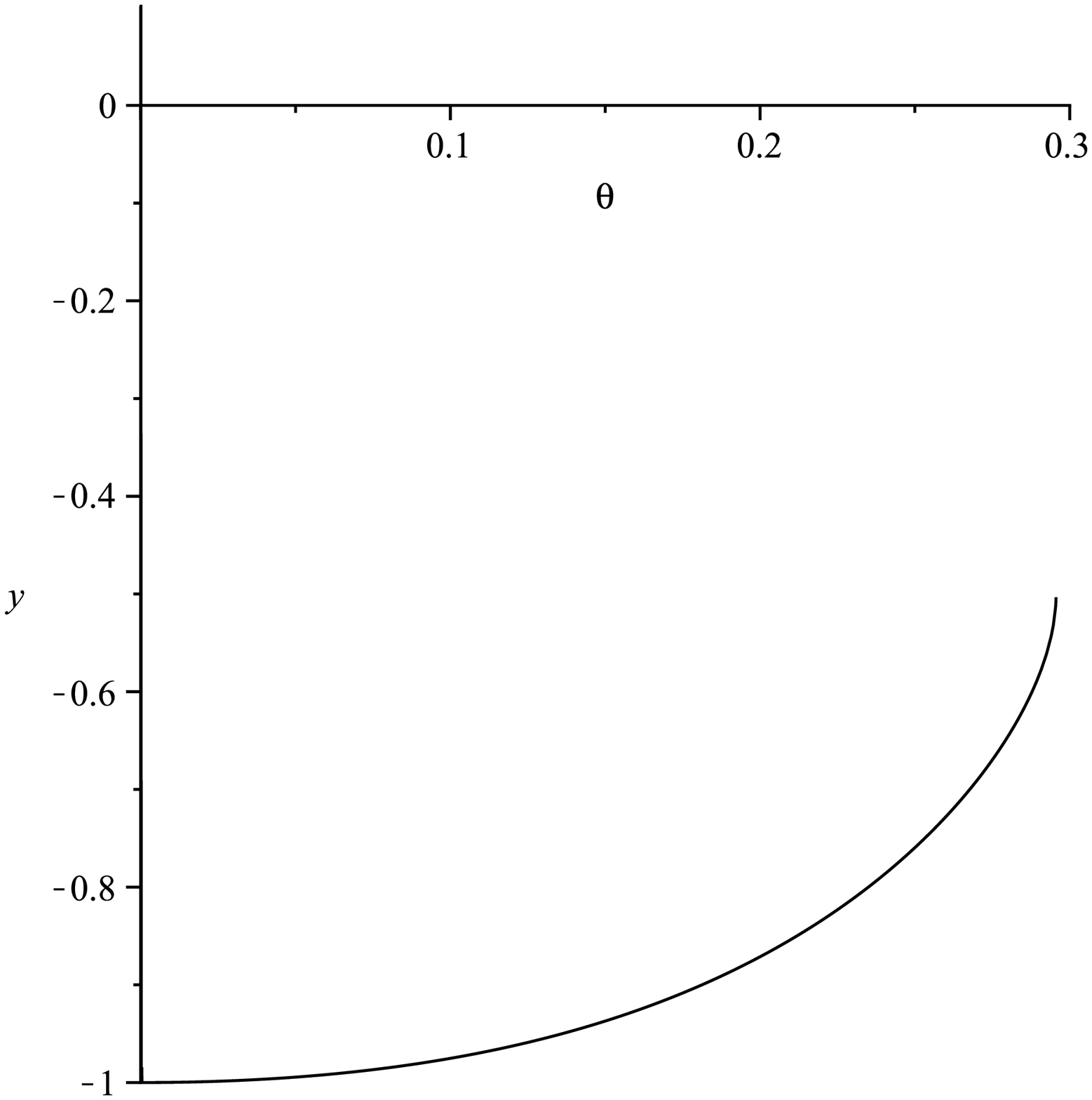}\ \ \ \
  \includegraphics[width=6.5cm]{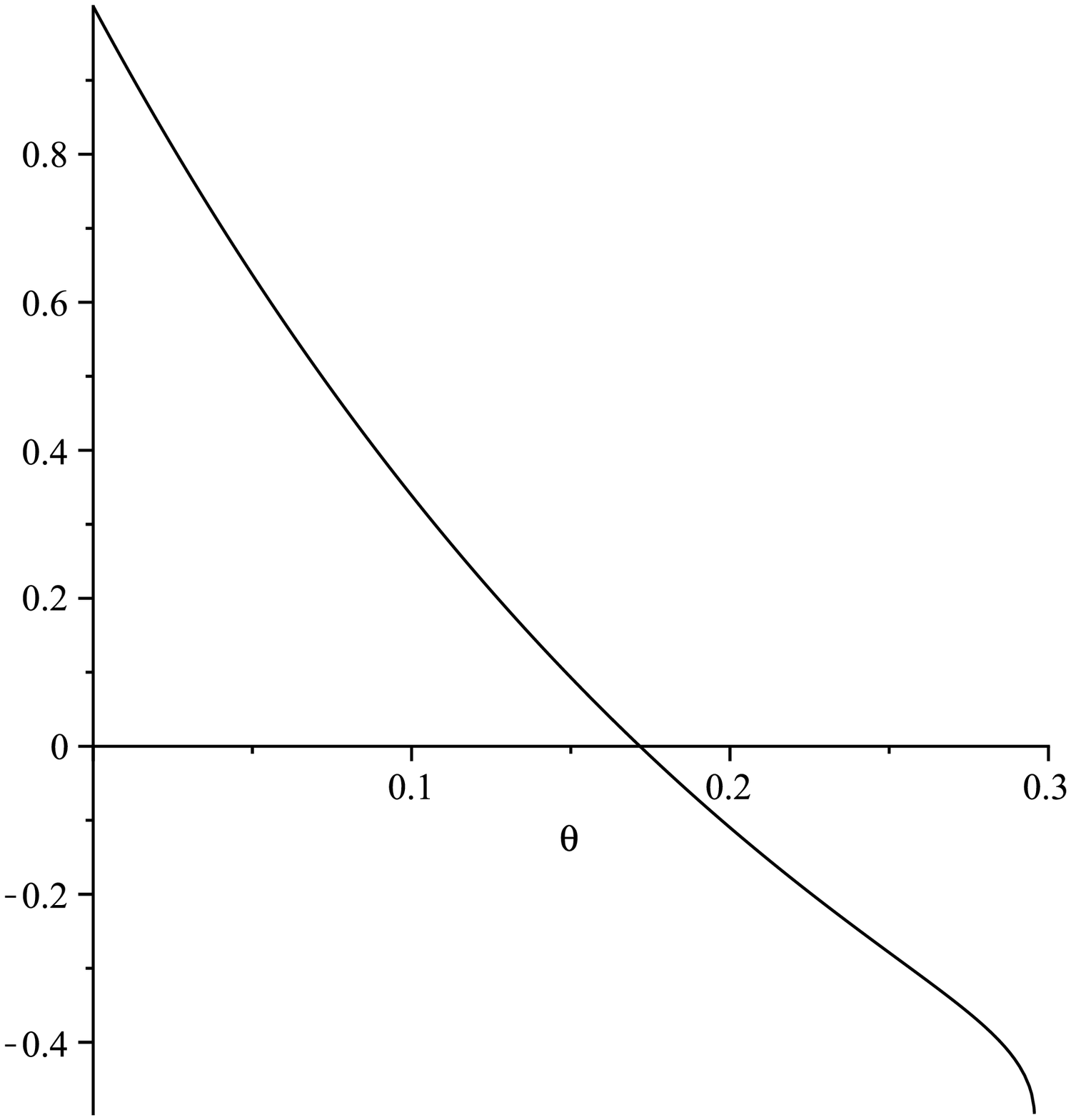}
  \end{center}
  \caption{The graphs of functions $\eta_4(\theta)$ (left) and $\eta_5(\theta)$ (right).}\label{ny5ny6}
   \end{figure}
   \begin{figure}
   \begin{center}
    \includegraphics[width=6cm]{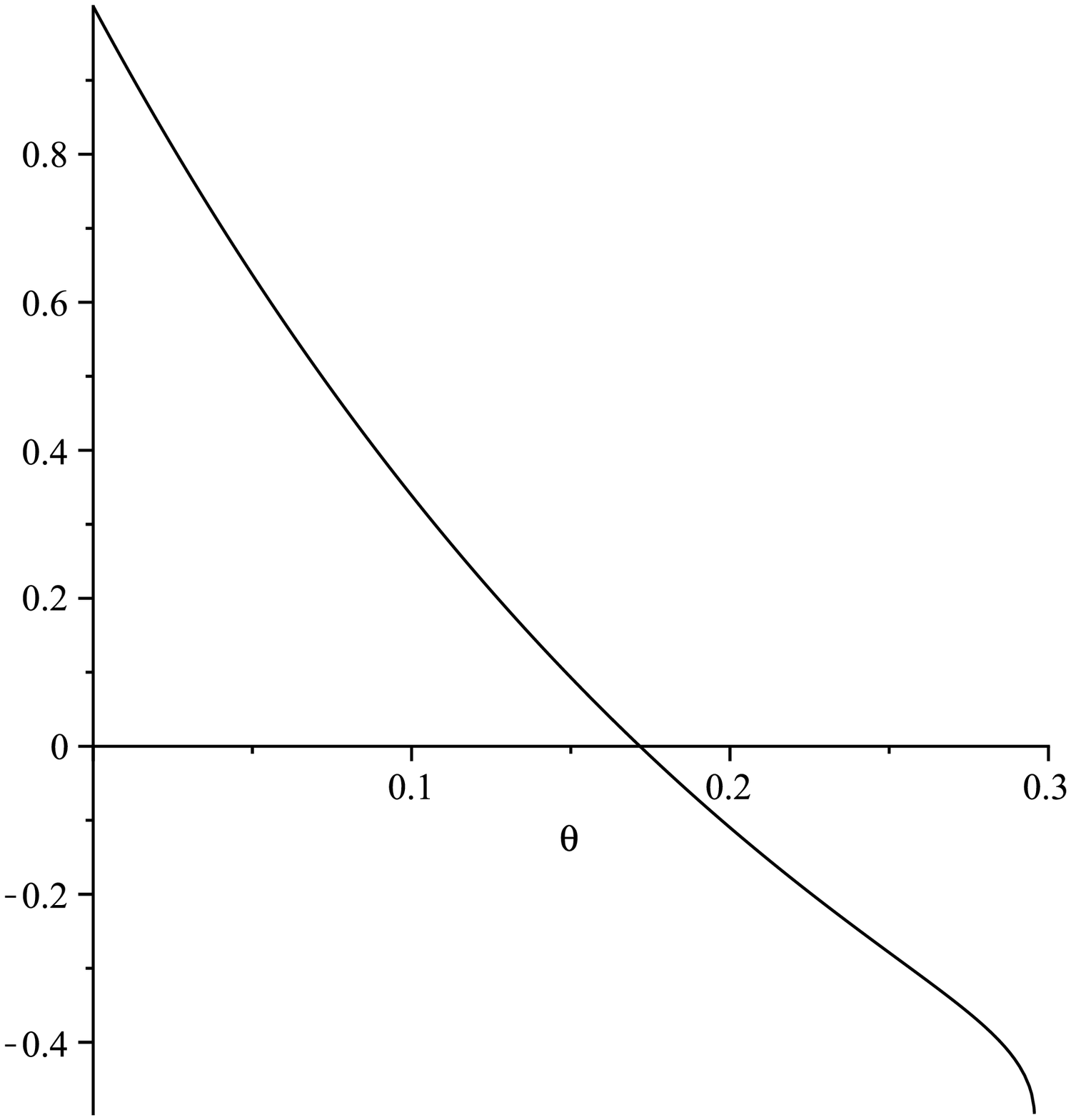}\ \ \ \
  \includegraphics[width=6cm]{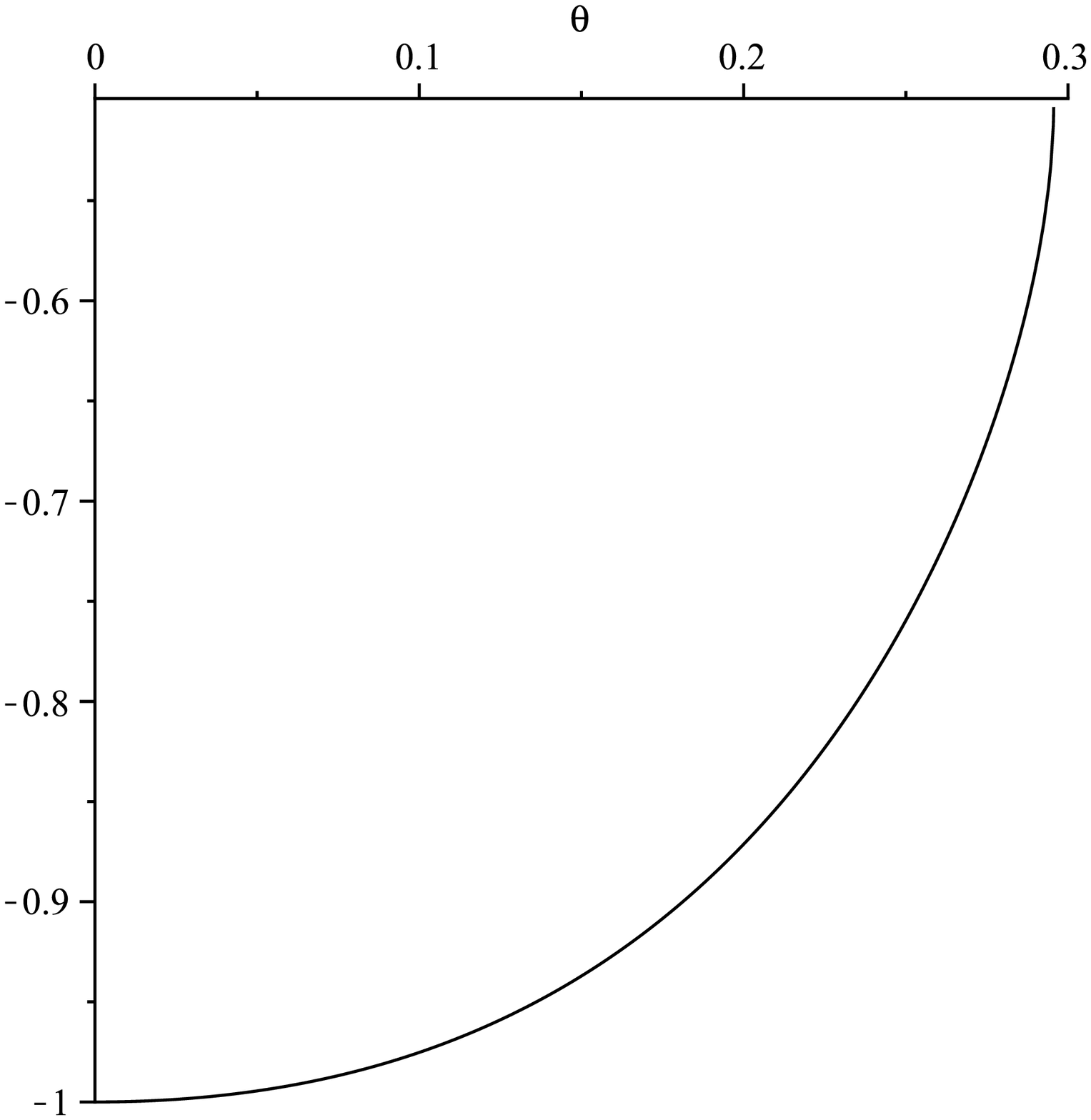}
  \end{center}
  \caption{The graph of functions $\eta_6(\theta)$ (left) and $\eta_7(\theta)$ (right).}\label{ny7}
   \end{figure}  
     \begin{figure}
      \begin{center}
     \includegraphics[width=9.3cm]{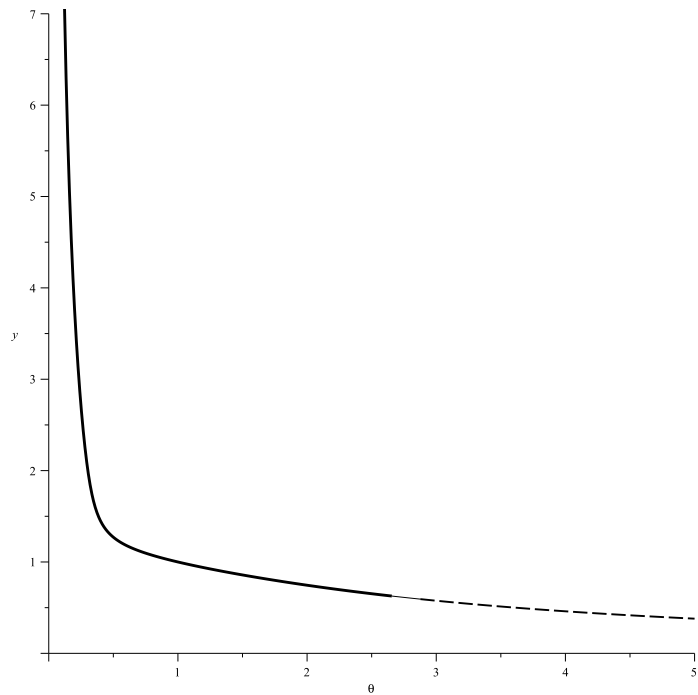}
     \end{center}
      \caption{The graph of $y_1(\theta)$. The bold curve corresponds to region of the function where corresponding TISGM is extreme. The dashed bold curve corresponds to region of the function where corresponding TISGM is non-extreme. The gap between the two types of curves are given by the thin curve. The length of gap is $\approx 0.22$. }\label{ney1}
    \end{figure}
      \begin{figure}
      \begin{center}
      \includegraphics[width=9.1cm]{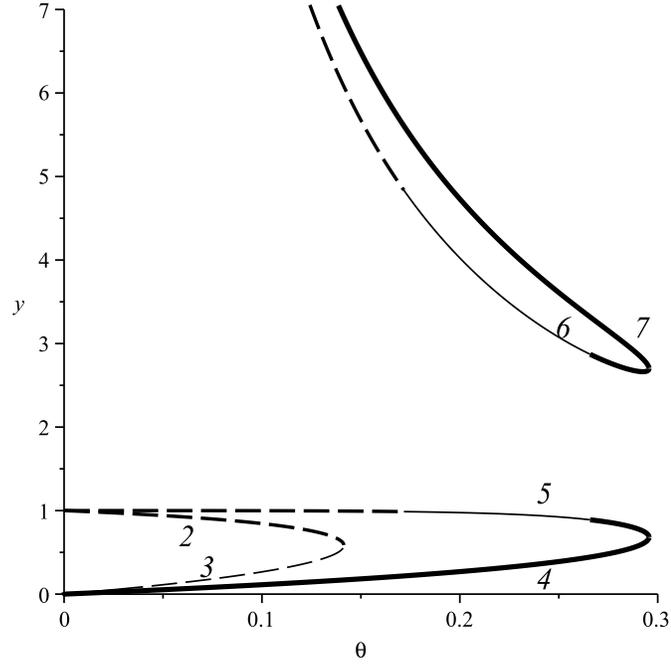}
     \end{center}
     \caption{The graphs of the functions $y_i(\theta)$, $i=2,3,4,5,6,7$.  The bold curves correspond to regions of the functions where the corresponding TISGM is extreme. The dashed  curves correspond to regions of the functions where corresponding TISGMs are non-extreme. The gap between the two types of curves are given by the thin curves in the graphs of $y_5$ and $y_6$. The length of each gap is $\approx 0.09.$ }\label{ney6y7}
    \end{figure} 
      \begin{figure}
      \begin{center}
      \includegraphics[width=5.7cm]{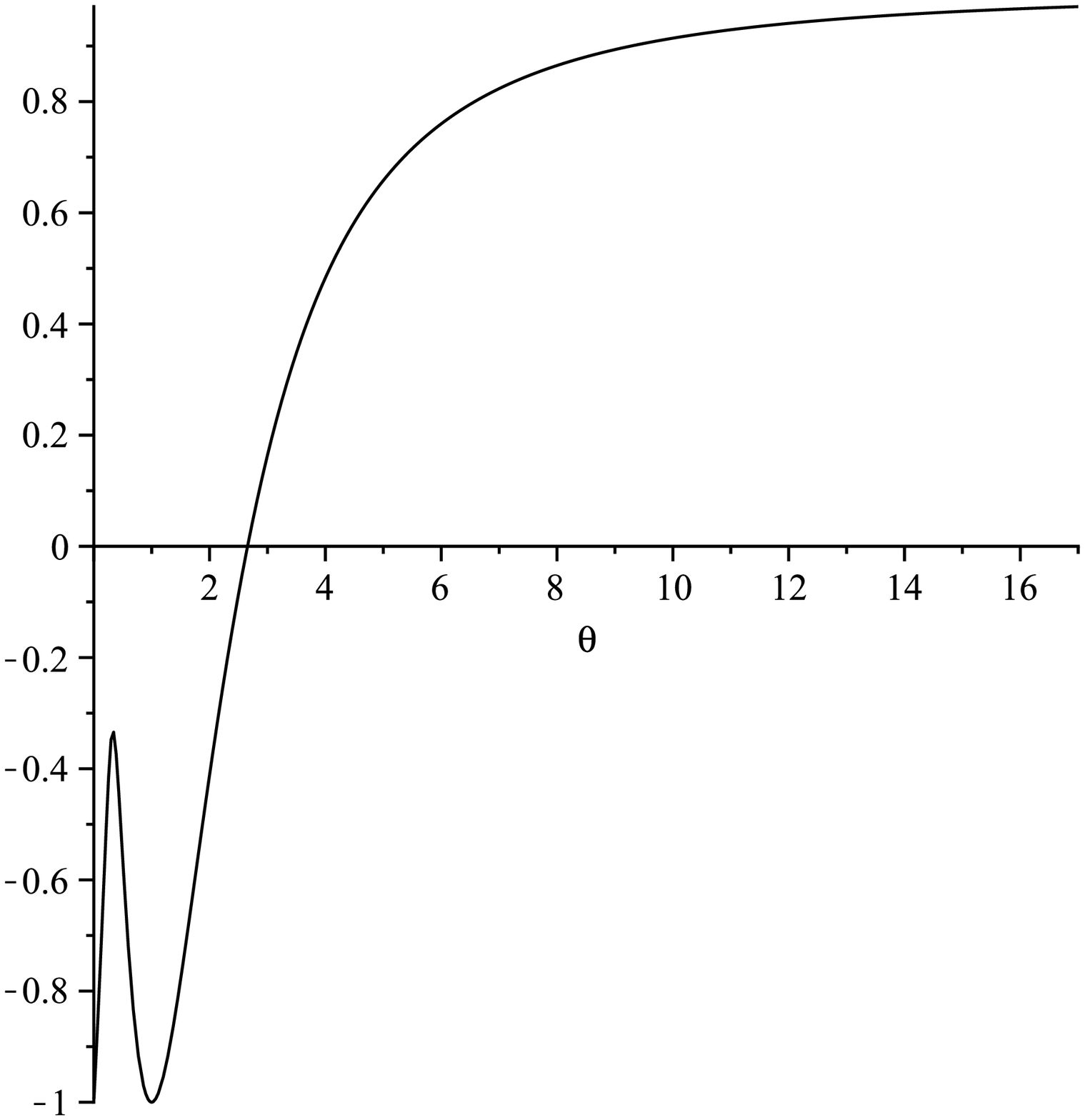} \ \ \ \
       \includegraphics[width=5.7cm]{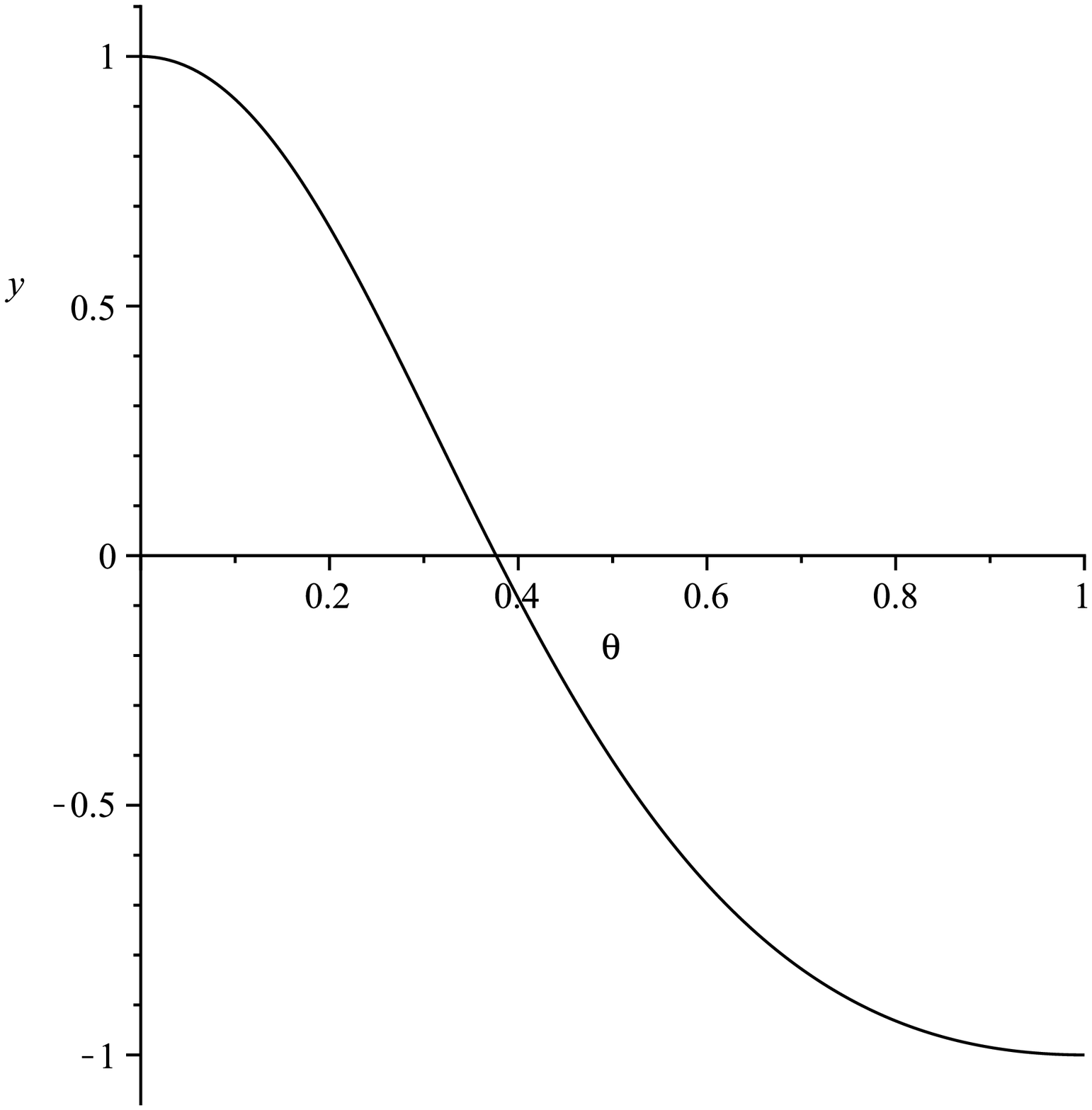}
     \end{center}
     \caption{The graph of the function $U_1(\theta)$ (left) and $U_1(1/\theta)$ (right). }\label{ey1}
     \end{figure} 
     \begin{figure}
       \begin{center}
       \includegraphics[width=6.5cm]{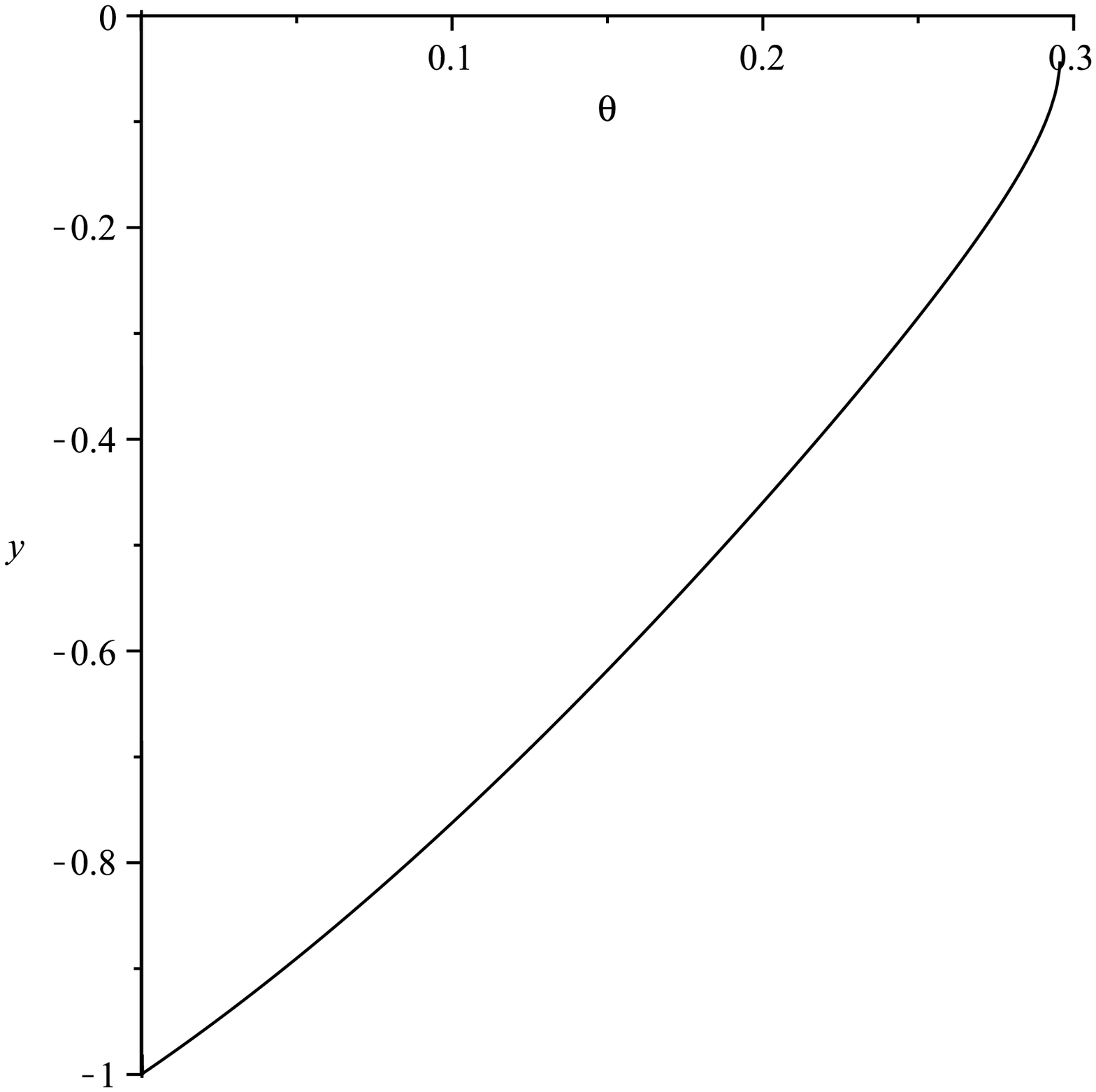} \ \
        \includegraphics[width=6.5cm]{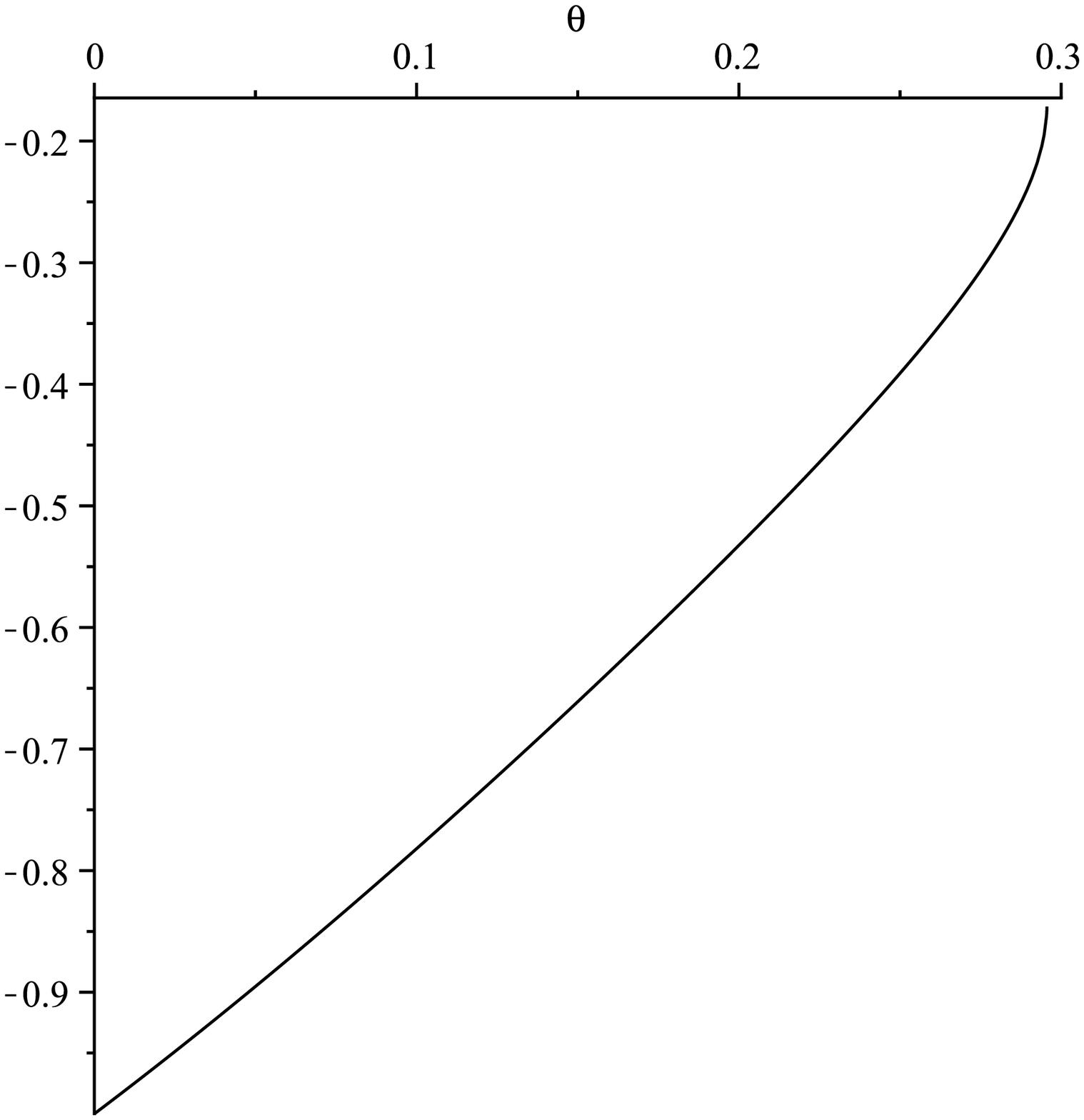} 
      \end{center}
      \caption{The graph of the function $U_4(\theta)$ (left) and $U_7(\theta)$ (right).} \label{ey4}
      \end{figure} 
       \begin{figure}
         \begin{center}
         \includegraphics[width=6.5cm]{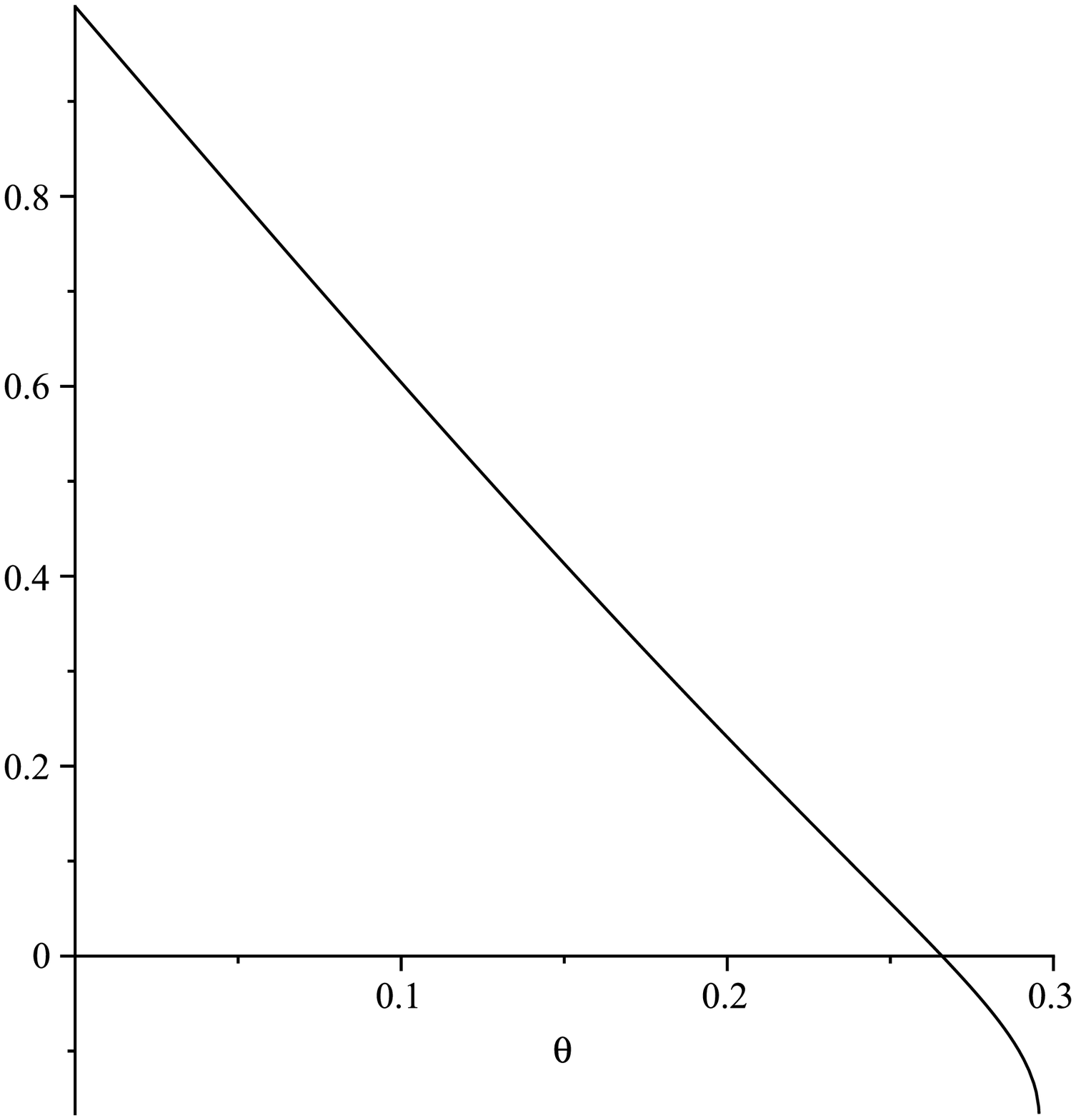} \ \
          \includegraphics[width=6.5cm]{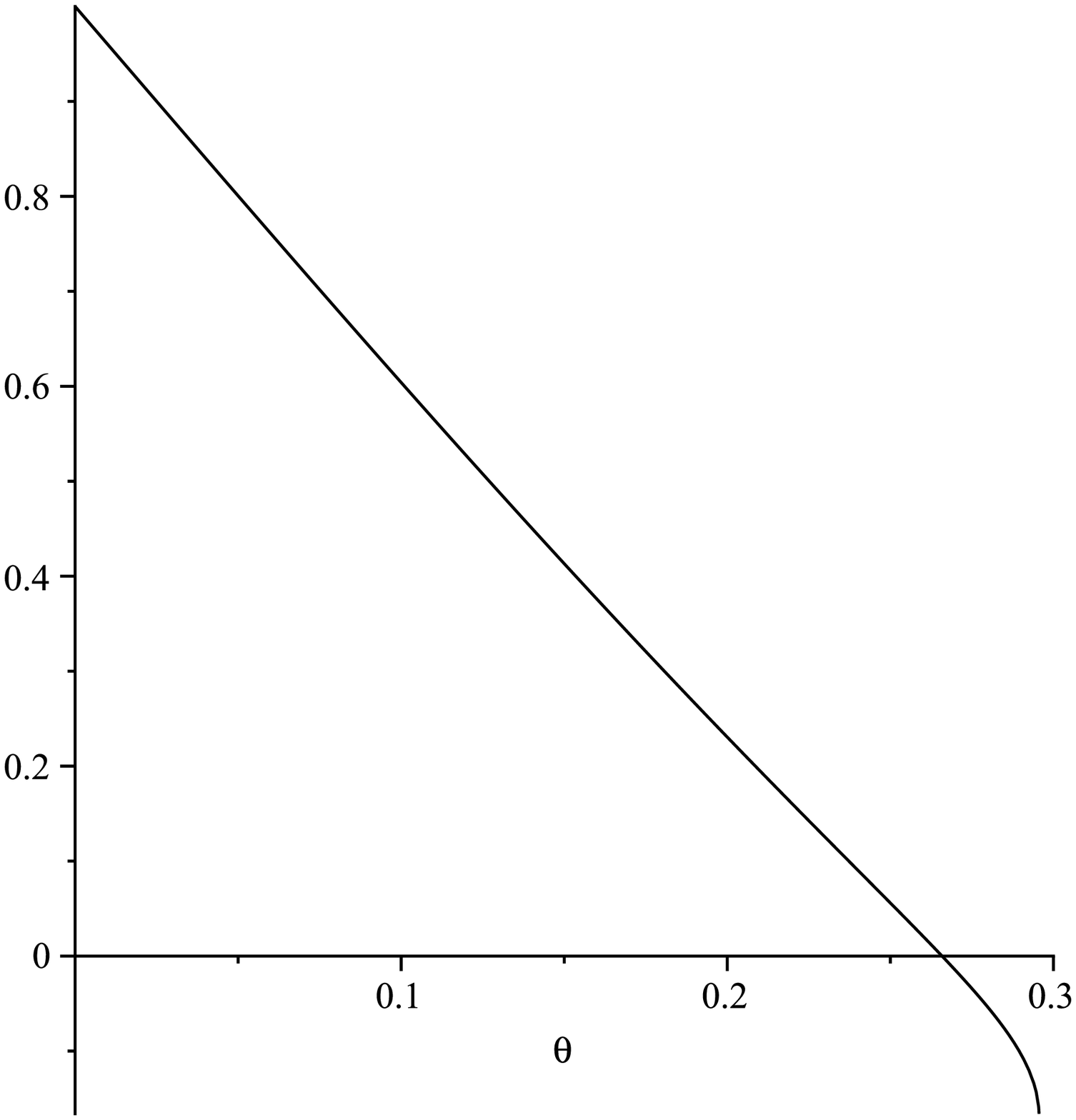} 
        \end{center}
        \caption{The graph of the function $U_5(\theta)$ (left) and $U_6(\theta)$ (right).} \label{ey5}
        \end{figure}


\begin{thebibliography}{99}
 
 \bibitem{BK} A. Bovier, C. K\"ulske, \textit{A rigorous renormalization group method for interfaces in random media},  Rev. Math. Phys. \textbf{6}(3) (1994), 413--496.
 
 \bibitem{BK1} A. Bovier, C. K\"ulske, \textit{There are no nice interfaces in $2+1$ dimensional SOS-models in random media}, 
 J. Stat. Phys. \textbf{83} (1996),  751--759.
 
 \bibitem{Ge} H.O. Georgii, \textit{Gibbs Measures and Phase Transitions},  Second edition. de Gruyter Studies in Mathematics, 9. Walter de Gruyter, Berlin, 2011.
 
 \bibitem{Kel} F.P. Kelly, \textit{Stochastic models of computer communication systems.
 With discussion},  J. Roy. Statist. Soc. Ser. B \textbf{47} (1985), 379--395;
 415--428.
 
 \bibitem{Ke} H. Kesten, B.P. Stigum, \textit{Additional limit theorem for indecomposable multi-dimensional Galton-Watson processes}, Ann. Math. Statist. \textbf{37} (1966), 1463--1481.
 
 \bibitem{KRK} C. K\"ulske, U.A. Rozikov, R.M. Khakimov, \textit{Description of the translation-invariant splitting Gibbs measures for the Potts model on a Cayley tree}.  Jour. Stat. Phys. \textbf{156}(1) (2014), 189--200.
 
 \bibitem{KR} C. Kuelske, U. A. Rozikov, \textit{Fuzzy transformations and extremality of Gibbs measures for the Potts model on a Cayley tree}.
 arXiv:1403.5775.
 
 \bibitem{MSW} F. Martinelli, A. Sinclair, D. Weitz, \textit{Fast mixing for independent sets, coloring and other models on
 trees}. Random Structures and Algoritms, \textbf{31} (2007), 134-172.
 
 \bibitem{Maz} A.E. Mazel, Yu.M. Suhov, \textit{Random surfaces with two-sided constraints: an application of the theory of dominant ground states}, J. Statist. Phys. \textbf{64} (1991), 111--134.
 
 \bibitem{Mos2} E. Mossel, Y. Peres, \textit{Information flow on trees}, Ann. Appl. Probab.  \textbf{13}(3) (2003), 817--844.
 
 \bibitem{Mos} E. Mossel, \textit{Survey: Information Flow on Trees}.  Graphs, morphisms and statistical physics,  155--170, DIMACS Ser. Discrete Math. Theoret. Comput. Sci., 63, Amer. Math. Soc., Providence, RI, 2004.
 
 \bibitem{Ra} K. Ramanan, A. Sengupta, I. Ziedins, P.
 Mitra, \textit{Markov random field models of multicasting in
 tree networks},  Adv. Appl. Probab. \textbf{34} (2002), 58--84.
 
 \bibitem{Ro12} U.A. Rozikov, Yu.M. Suhov, \textit{Gibbs measures for SOS model on a Cayley tree},  Infin. Dimens. Anal.
 Quantum Probab. Relat. Top. \textbf{9}(3) (2006), 471--488.
 
 \bibitem{Ro13} U.A. Rozikov, Sh.A. Shoyusupov, \textit{Gibbs measures for the SOS model with four states on a Cayley
 tree},  Theor. Math. Phys. \textbf{149}(1) (2006), 1312--1323.
 
  \bibitem{Ro} U.A.  Rozikov, \textit{Gibbs measures on Cayley trees}.  World Sci. Publ. Singapore. 2013.
 
  \bibitem{Sly11} A. Sly, \textit{Reconstruction for the Potts model}. Ann. Probab.  \textbf{39}  (2011),  1365--1406.
 
   \end{thebibliography}
\end{document}